\documentclass[a4paper,UKenglish,cleveref,autoref,thm-restate,authorcolumns]{lipics-v2019}

\pdfoutput=1

\newcommand{\ignore}[1]{}

\newcommand{\ceil}[1]{\left\lceil #1 \right\rceil}

\newcommand{\poly}{{\operatorname{poly}}}

\newcommand{\dist}{\operatorname{dist}}

\newcommand{\leaf}{\operatorname{leaf}}
\newcommand{\inode}{\operatorname{inode}}

\newcommand{\reduce}{\mathsf{Reduce}}

\newcommand{\maximalpaths}{\mathsf{MaximalPaths}}
\newcommand{\Pin}{\mathcal{P}_{\operatorname{in}}}
\newcommand{\Pout}{\mathcal{P}_{\operatorname{out}}}
\newcommand{\Pactive}{\mathcal{P}_{\operatorname{a}}}

\newcommand{\idle}{\mathsf{Idle}}
\newcommand{\aactive}{\mathsf{Active}}
\newcommand{\dead}{\mathsf{Dead}}

\newcommand{\pred}{\operatorname{pred}}

\usepackage{xspace}

\title{Streaming Complexity of Spanning Tree Computation} 

\author{Yi-Jun Chang}{ETH Z\"{u}rich, Switzerland}{yi-jun.chang@eth-its.ethz.ch}{}{}
\author{Mart\'{i}n Farach-Colton}{Rutgers University, USA}{farach@cs.rutgers.edu}{}{This research was supported in part by NFS grants CSR-1938180, CCF-1715777, and CCF-1724745.}
\author{Tsan-Sheng Hsu}{Academia Sinica, Taiwan}{tshsu@iis.sinica.edu.tw}{}{This research was supported in part by the Ministry of Science and Technology of Taiwan under contract MOST Grant 108-2221-E-001-011-MY3.}
\author{Meng-Tsung Tsai}{National Chiao Tung University, Taiwan}{mtsai@cs.nctu.edu.tw}{}{This research was supported in part by the Ministry of Science and Technology of Taiwan under contract MOST grant 107-2218-E-009-026-MY3.}

\authorrunning{Y.-J. Chang et al.}

\Copyright{Yi-Jun Chang, Mart\'{i}n Farach-Colton, Tsan-Sheng Hsu, and Meng-Tsung Tsai}



\category{} 

\relatedversion{} 

\supplement{}

\acknowledgements{We  
thank the anonymous reviewers for their helpful comments, and Eric Allender and Meng Li for their insightful discussions.}

\nolinenumbers 

\hideLIPIcs  

\EventEditors{Christophe Paul and Markus Bl\"{a}ser}
\EventNoEds{2}
\EventLongTitle{37th International Symposium on Theoretical Aspects of Computer Science (STACS 2020)}
\EventShortTitle{STACS 2020}
\EventAcronym{STACS}
\EventYear{2020}
\EventDate{March 10--13, 2020}
\EventLocation{Montpellier, France}
\EventLogo{}
\SeriesVolume{154}
\ArticleNo{30}

\renewcommand{\paragraph}[1]{\smallskip {\noindent \textit{\bf #1}}}

\newcommand{\metaproof}[2]{
\begin{proof}#2\end{proof}
}

\begin{document}

\maketitle

\begin{abstract}

The semi-streaming model is a variant of the streaming model frequently used for the computation of graph problems.  It allows the edges of an $n$-node input graph to be read sequentially in $p$ passes using $\tilde{O}(n)$ space. If the list of edges includes deletions, then the model is called the turnstile model; otherwise it is called the insertion-only model. In both models, some graph problems, such as spanning trees, $k$-connectivity, densest subgraph, degeneracy, cut-sparsifier, and $(\Delta+1)$-coloring, can be exactly solved or $(1+\varepsilon)$-approximated in a single pass; while other graph problems, such as triangle detection and unweighted all-pairs shortest paths, are known to require $\tilde{\Omega}(n)$ passes to compute. For many fundamental graph problems, the tractability in these models is open. In this paper, we study the tractability of computing some standard spanning trees, including BFS, DFS, and maximum-leaf spanning trees. 

Our results, in both the insertion-only and the turnstile models, are as follows.

\begin{description}

\item [Maximum-Leaf Spanning Trees:] This problem is known to be APX-complete with inapproximability constant $\rho \in [245/244, 2)$. By constructing an \emph{$\varepsilon$-MLST sparsifier}, we show that for every constant $\varepsilon > 0$,  MLST can be approximated in a single pass to within a factor of $1+\varepsilon$ w.h.p. (albeit in super-polynomial time for $\varepsilon \le \rho-1$ assuming $\mathrm{P} \ne \mathrm{NP}$) and can be approximated in polynomial time in a single pass to within a factor of $\rho_n+\varepsilon$ w.h.p., where $\rho_n$ is the supremum constant that MLST cannot be approximated to within using polynomial time and $\tilde{O}(n)$ space. In the insertion-only model, these algorithms can be deterministic. 

\item [BFS Trees:] It is known that BFS trees require $\omega(1)$ passes to compute, but the na\"{i}ve approach needs $O(n)$ passes. We devise a new randomized algorithm that reduces the pass complexity to $O(\sqrt{n})$, 
and it offers a smooth tradeoff between pass complexity and space usage.  This gives a polynomial separation between single-source and all-pairs shortest paths for unweighted graphs. 

\item [DFS Trees:] It is unknown whether DFS trees require more than one pass. The current best algorithm by Khan and Mehta {[}STACS 2019{]} takes $\tilde{O}(h)$ passes, where $h$ is the height of computed DFS trees. Note that $h$ can be as large as $\Omega(m/n)$ for $n$-node $m$-edge graphs. 
Our contribution is twofold. 
First, we provide a simple alternative proof of this result, via a new connection to sparse certificates for $k$-node-connectivity.
Second, we present a randomized algorithm that reduces the pass complexity to $O(\sqrt{n})$, 
and it also offers a smooth tradeoff between pass complexity and space usage.

\end{description}

\end{abstract}


\section{Introduction}\label{sec:intro}

Spanning trees are  critical components of graph algorithms, from depth-first search trees (DFS) for finding articulation points and bridges~\cite{Tarjan74}, computing $st$-numbering~\cite{Even76}, chain decomposition~\cite{Schmidt13}, and coloring signed graphs~\cite{Fleiner16}, to breadth-first search trees (BFS) for finding separators~\cite{Lipton79}, computing sparse certificates of $k$-node-connectivity~\cite{Cheriyan91,EppsteinGIN97},
approximating diameters~\cite{CorneilDK03,RodittyW13}, and characterizing AT-free graphs~\cite{Beisegel18}, and to maximum-leaf spanning trees (MLST) for connected dominating sets~\cite{LuR98,Solis-Oba17} and connected maximum cuts~\cite{HajiaghayiKMPS15,GandhiHKPS18}.

In the semi-streaming model, the tractability of spanning tree computation, except arbitrary spanning trees~\cite{AhnGM12,SunW15,NelsonY19}, 
is less studied. The \emph{semi-streaming} model~\cite{Muthu05,AhnGM12} is a variation of streaming model frequently used for the computation of graph problems. It allows the edges of an $n$-node input graph to be read sequentially in $p$ passes using $\tilde{O}(n)$\footnote{We write $\tilde{O}(k)$ to denote $O(k\, \poly \log n)$ or $O(k/\poly \log n)$ where $n$ is the number of nodes in the input graph. Similarly, $\tilde{\Omega}(k)$ denotes $\Omega(k\,\poly \log n)$ or $\Omega(k/\poly \log n)$.} space. If the list of edges includes deletions, then the model is called the turnstile model; otherwise it is called the insertion-only model. In both models, some graph problems, such as spanning trees~\cite{AhnGM12}, $k$-connectivity~\cite{GuhaMT15}, densest subgraph~\cite{McGregorTVV15}, degeneracy~\cite{Farach-ColtonT16}, cut-sparsifier~\cite{KapralovLMMS17}, and $(\Delta+1)$-coloring~\cite{AssadiCK19}, can be exactly solved or $(1+\varepsilon)$-approximated in a single pass, while other graph problems, such as triangle detection and unweighted all-pairs shortest paths~\cite{BravermanOV13}, are known to require $\tilde{\Omega}(n)$ passes to compute. For many fundamental graph problems, e.g., standard spanning trees, the tractability in these models is open. BFS computation is known to require $\omega(1)$ passes~\cite{Feigenbaum08}, but only the naive $O(n)$-pass algorithm is known. It is unknown whether DFS computation requires more than one passes~\cite{Farach-Colton15,khan2019depth}, but the current best algorithm needs $\tilde{O}(h)$ passes~\cite{khan2019depth} where $h$ is the height of the computed DFS trees, so $h = O(n)$ for dense graphs. The tractability of maximum-leaf spanning trees (MLST) is unknown even allowing  $O(n^2)$ space, since it is APX-complete~\cite{LuR92,Galbiati94}.

Due to the lack of efficient streaming algorithms for spanning tree computation, for some graph problems that are traditionally solved using spanning trees, such as finding articulation points and bridges,  people had to look for alternative methods when designing streaming algorithms for these problems~\cite{Feigenbaum05,Farach-Colton15}.
The alternative methods, even if they are based on known results in graph theory, may still involve the design of new streaming algorithms. For the  problems mentioned above, the alternative methods use newly-designed sparse connectivity certificates~\cite{EppsteinGIN97,GuhaMT15} that are easily computable in the semi-streaming model,
rather than the classical one due to Nagamochi and Ibaraki~\cite{Nagamochi92}. Hence establishing the hardness of spanning tree computation helps to explain the need of the alternative methods.

In this paper, we study the tractability of computing standard spanning trees for connected simple undirected graphs, including BFS trees, DFS trees, and MLST. Unless otherwise stated, our upper bounds work in the turnstile model (and hence also in the insertion-only model), and our lower bounds hold for the insertion-only model (and hence also in the turnstile model). The space upper and lower bounds are in bits. Our results are as follows.

    \paragraph{Maximum-Leaf Spanning Trees:} 
    We show, by constructing an \emph{$\varepsilon$-MLST sparsifier} (\cref{thm:mlst-sparsifier}), that for every constant $\varepsilon > 0$,  MLST can be approximated in a single pass to within a factor of $1+\varepsilon$ w.h.p.\footnote{W.h.p. means with probability $1 - 1/n^{\Omega(1)}$.} (albeit in super-polynomial time for $\varepsilon \le \rho-1$ since it is APX-complete~\cite{LuR92,Galbiati94} with inapproximability constant $\rho \in [245/244, 2)$~\cite{ChlebikC08}) and can be approximated in polynomial time in a single pass to within a factor of $\rho_n+\varepsilon$ w.h.p., where $\rho_n$ is the supremum constant that MLST cannot be approximated to within using polynomial time and $\tilde{O}(n)$ space. In the insertion-only model, these algorithms are deterministic. We also show a complementary hardness result (\cref{thm:lower-mlst}) that for every $k \in [1, (n-5)/4]$, to approximate MLST to within an additive error $k$, any single-pass randomized streaming algorithm that succeeds with probability at least $2/3$ requires $\Omega(n^2/k^2)$ bits.  This hardness result excludes the possibility to have a single-pass semi-streaming algorithm to approximate MLST to within an additive error $n^{1/2-\Omega(1)}$.
    Our results for MLST shows that intractability in the sequential computation model (i.e., Turing machine) does not imply intractability in the semi-streaming model.
    
    Our algorithms rely on a new sparse certificate, the \emph{$\varepsilon$-MLST sparsifier}, defined as follows.  Let $G$ be an $n$-node $m$-edge connected simple undirected graph.  Then for any given constant $\varepsilon > 0$, $H$ is an \emph{$\varepsilon$-MLST sparsifier} if it is a connected spanning subgraph of $G$ with $|E(H)| \le f(\varepsilon) |V(G)|$ and $\leaf(H) \ge (1-\varepsilon) \leaf(G)$, where $\leaf(G)$ denotes the maximum number of leaves (i.e. nodes of degree one) that any spanning tree of $G$ can have and $f$ is some function independent of $n$. 
    We show that an $\varepsilon$-MLST sparsifier can be constructed efficiently in the semi-streaming model.
    
    \begin{theorem}\label{thm:main-mlst}
    In the turnstile model, for every constant $\varepsilon > 0$, there exists a randomized algorithm that can find an $\varepsilon$-MLST sparsifier with probability $1-1/n^{\Omega(1)}$ using a single pass, $\tilde{O}(f(\varepsilon) n)$ space, and $\tilde{O}(n+m)$ time, and in the insertion-only model a deterministic algorithm that uses a single pass, $\tilde{O}(f(\varepsilon) n)$ space, and $O(n+m)$ time. 
    \end{theorem}
    
    Combining \cref{thm:main-mlst} with any polynomial-time RAM algorithms for MLST that uses $\tilde{O}(n+m)$ space, e.g,~\cite{LuR92,LuR98,Solis-Oba17}, we obtain the following result.
    
    \begin{corollary}\label{cor:apx-mlst}
    In the turnstile model, for every constant $\varepsilon > 0$, there exists a randomized algorithm that can approximate $MLST$ for any $n$-node connected simple undirected graph with probability $1-1/n^{\Omega(1)}$ to within a factor of $\rho_n+\varepsilon$ using a single pass, $\tilde{O}(f(\varepsilon) n)$ space, and polynomial time, where $\rho_n$ is the supremum constant that MLST cannot be approximated to within using polynomial time and $\tilde{O}(n)$ space, and in the insertion-only model a deterministic algorithm that uses a single pass, $\tilde{O}(f(\varepsilon) n)$ space, and polynomial time.    
    \end{corollary}
    
    Using \cref{cor:apx-mlst}, 
    we show that  approximate connected maximum cut can be computed in a single pass using $\tilde{O}(n)$ space for unweighted regular graphs (\cref{cor:cmcut}).  
    
    \paragraph{BFS Trees:} It is known that BFS trees require $\omega(1)$ passes to compute~\cite{Feigenbaum08}, but the naive approach needs $O(n)$ passes. We devise a randomized algorithm that reduces the pass complexity to $O(\sqrt{n})$ w.h.p., and give a smooth tradeoff between pass complexity and space usage.
    
    \begin{theorem}\label{thm:main-bfs}
    In the turnstile model, for each $p \in [1, \sqrt{n}]$, there exists a randomized algorithm that can compute a BFS tree for any $n$-node connected simple undirected graph with probability $1-1/n^{\Omega(1)}$ in $p$ passes using $\tilde{O}((n/p)^2)$ space, and in the insertion-only model a deterministic algorithm that uses $\tilde{O}(n^2/p)$ space. 
    \end{theorem}
    
    This gives a polynomial separation between single-source and all-pairs shortest paths for unweighted graphs
    because any randomized semi-streaming algorithm that computes unweighted all-pairs shortest paths with probability at least $2/3$ requires $\tilde{\Omega}(n)$ passes. 
    
    We extend \cref{thm:main-bfs} and obtain that multiple BFS trees, each starting from a unique source node, can be computed more efficiently in pass complexity in a batch than individually (see \cref{thm:c-bfs}). We show that this batched BFS has applications to computing a $1.5$-approximation of diameters for unweighted graphs (\cref{thm:diam-apx}) and a $2$-approximation of Steiner trees for unweighted graphs (\cref{cor:Steiner}).

    \paragraph{DFS Trees:} It is unknown whether DFS trees require more than one passes~\cite{Farach-Colton15,khan2019depth}, but the current best algorithm needs $\tilde{O}(h)$ passes due to Khan and Mehta~\cite{khan2019depth}, where $h$ is the height of computed DFS trees. We devise a randomized algorithm that has pass complexity $O(\sqrt{n})$ w.h.p., and give a smooth tradeoff between pass complexity and space usage.
    
    \begin{theorem}\label{thm:main-dfs}
    In the turnstile model, for each $p \in [1, \sqrt{n}]$, there exists a randomized algorithm that can compute a DFS tree for any $n$-node connected simple undirected graph with probability $1-1/n^{\Omega(1)}$ in $p$ passes that uses $\tilde{O}(n^3/p^4)$ space, and in the insertion-only model a deterministic algorithm that uses $\tilde{O}(n^2/p^2)$ space.
    \end{theorem}
    
     For dense graphs, our algorithms improves upon the current best algorithms for DFS due to Khan and Mehta~\cite{khan2019depth} which  needs $\Omega(m/n)$ passes for $n$-node $m$-edge graphs in the worst case because of the existence of $(m/n)$-cores, where a $k$-core is a maximal connected subgraph in which every node has at least $k$ neighboring nodes in the subgraph.

\subsection{Technical Overview}

    \paragraph{Maximum-Leaf Spanning Trees:} We construct an $\varepsilon$-MLST sparsifier by a new result that complements Kleitman and West's lower bounds on the maximum number of leaves for graphs with minimum degree $\delta \ge 3$~\cite{Kleitman91}. The lower bounds are: if a connected simple undirected graph $G$ has minimum degree $\delta$ for some sufficiently large $\delta$, then $\leaf(G) \ge (1- (2.5\ln \delta)/\delta)|V(G)|$ and the leading constant can be larger for $\delta \in \{3, 4\}$. Our complementary result (\cref{lem:dead-leaf}), without the restriction on the minimum degree, is: any connected simple undirected graph $G$, except the singleton graph, has 
    \begin{equation}
        \leaf(G) \ge \frac{1}{10}(|V(G)|-\inode(G)), \label{eqn:leafy}
    \end{equation}
    where $\inode(G)$ denotes the number of nodes whose degree is two and whose neighbors both have degree two. \cref{eqn:leafy} implies that, if one can find a connected spanning subgraph $H$ of $G$ so that $|\leaf(G) - \leaf(H)| \le \varepsilon (V(G) - \inode(G))$, then one gets an $(10\varepsilon)$-MLST sparsifier.

    Our sparsification technique is general enough to obtain a $(t+\varepsilon)$-approximation for MLST in a single pass using $\tilde{O}(n)$ space by combining any $t$-approximation $\tilde{O}(n)$-space RAM algorithm for MLST with our $\varepsilon$-MLST sparsifier. On the other hand, since in linear time one can find an $\varepsilon$-MLST sparsifier of $O(n)$ edges, any $t$-approximation RAM algorithm for MLST with time complexity $O(f(n, m))$ can be reduced to $O(f(n, n) + n + m)$ if a small sacrifice on approximation ratio is allowed. This reduces the time complexity of RAM algorithms for MLST that need superlinear time on the number of edges, such as the local search approach from $O(m^k n^{k+2})$ for $k \ge 1$ to $O(n^{2k+2})$ and the leafy forest approach from $O((m+n)\alpha(n))$ to $O(m+n\alpha(n))$, both due to Lu and Ravi~\cite{LuR92,LuR98}.

    \paragraph{BFS Trees:} We present a simple deterministic algorithm attaining a smooth tradeoff between pass complexity and space usage.
      In particular,  in the insertion-only semi-streaming model, the algorithm finishes in $O(n / \poly \log n)$ passes.
    The algorithm is based on an observation that the sum of degrees of nodes in any root-to-leaf path of a BFS tree is bounded by $O(n)$ (\cref{lem:degsum}).

    Our more efficient randomized algorithm (\cref{thm:main-bfs})
    constructs a BFS tree by combining the results of multiple instances of bounded-radius BFS.
    To reduce the space usage, the simulation of these bounded-radius BFS are assigned random starting times, and the algorithm only maintains the last three layers of each BFS tree.
    These ideas are borrowed from results on shortest paths computation in the parallel and the distributed settings~\cite{Elkin17,GhaffariL17,HuangNS17,UllmanY91}.

\paragraph{DFS Trees:} 
We present a  simple alternative proof of the result of Khan and Mehta~\cite{khan2019depth} that a DFS tree can be constructed in 
$\ceil{h/k}$ passes  using $\tilde{O}(nk)$ space, for any given parameter $k$, where $h$ is the height of the computed DFS tree. The new proof is based on the following connection between the DFS computation and the sparse certificates for $k$-node-connectivity. We show in Lemma~\ref{lem:simple-dfs-aux} that the first $k$ layers of \emph{any} DFS tree of a such a certificate $H$ can be extended to a DFS tree of the original graph $G$.

The proof of Theorem~\ref{thm:main-dfs} is based on the parallel DFS algorithm of Aggarwal and Anderson~\cite{Aggarwal1988}. 
In this paper we provide an efficient implementation of their algorithm in the streaming model, also via the sparse certificates for $k$-node-connectivity, which allows us to reduce the number of passes  by batch processing.

We note that in a related work, Ghaffari and Parter~\cite{GhaffariP17} showed that the parallel DFS algorithm of Aggarwal and Anderson can  be adapted to distributed setting. Specifically, they showed that DFS can be computed in the {\sf CONGEST} model in $\tilde{O}(\sqrt{Dn} + n^{3/4})$ rounds, where $D$ is the diameter of the graph.

\subsection{Paper Organization}

In \cref{sec:mlst}, we present how to construct an $\varepsilon$-MLST sparsifier and apply it to devise single-pass semi-streaming algorithms to approximate MLST to within a factor of $(1+\varepsilon)$ for every constant $\varepsilon > 0$. Then, in \cref{sect.bfs}, we show how to compute a BFS tree rooted at a given node by an $O(\sqrt{n})$-pass $\tilde{O}(n)$-space algorithm w.h.p. and its applications to computing approximate diameters and approximate Steiner trees. In Section~\ref{sect.dfs}, we have a similar result for computing DFS trees; that is, $O(\sqrt{n})$-pass $\tilde{O}(n)$-space algorithm that succeeds w.h.p. Lastly, we prove the claimed  single-pass lower bound in \cref{sec:lower}.

\newcommand{\expect}[1]{\mathbf{E}\left[#1\right]}

\section{Maximum-Leaf Spanning Trees}\label{sec:mlst}

In this section, we will show how to construct an $\varepsilon$-MLST sparsifier in the semi-streaming model; that is, proving \cref{thm:main-mlst}. We recall the notions defined in \cref{sec:intro} before proceeding to the results. By \emph{ignorable node}, we denote a node $x$ whose degree is two and whose neighbors $u$ and $v$ have degree two as well. Note that $u \ne v$ for simple graphs. Let $\leaf(G)$ be the maximum number of leaves (i.e. nodes of degree one) that a spanning tree of $G$ can have. Let $\inode(G)$ denote the number of ignorable nodes in $G$. Let $\deg_{G}(x)$ denote the degree of node $x$ in graph $G$. Let $S_k(G)$ denote any subgraph of $G$ so that $S_k(G)$ contains all nodes in $G$ and every node $x$ in $S_k(G)$ has degree $\deg_{S_k}(x) \ge \min\{\deg_{G}(x), k\}$. Let $T(G)$ be any spanning tree of a connected graph $G$.

We begin with a result that complements Kleitman and West's lower bounds on the number of leaves for graphs with minimum degree $\delta$ for any $\delta \ge 3$. Our lower bound does not rely on the degree constraint.  The constant $1/10$ in Lemma~\ref{lem:dead-leaf} may be improved, but the subsequent lemmata and theorems only require it to be $\Omega(1)$.

\begin{lemma}\label{lem:dead-leaf}
Every connected simple undirected graph $G$, except the singleton graph, has 
$$
    \leaf(G) \ge \frac{1}{10} (|V(G)| - \inode(G)).
$$

\end{lemma}
\metaproof{\cref{lem:dead-leaf}}{
Our proof is a generalization of the dead leaf argument due to Kleitman and West~\cite{Kleitman91}. Let $T$ be a tree rooted at $s$ with $N(s)$ as leaves for some arbitrary node $s \in G$ initially, where $N(s)$ denotes the neighbors of $s$, and then grow $T$ iteratively by a node expansion order, defined below. By expanding $T$ at node $x$, we mean to select a leaf node $x$ of $T$ and add all of $x$'s neighbors in $G \setminus T$, say $y_1, y_2, \ldots, y_d$, and their connecting edges, $(x, y_1), (x, y_2), \ldots, (x, y_d)$, to $T$. In this way, every node outside $T$ cannot be a neighbor of any non-leaf node in $T$. We say a leaf node in $T$ is \emph{dead} if it has no neighbor in $G \setminus T$. 
Let $(\Delta n)_i$ denote the number of non-ignorable nodes in $G$ that joins $T$ while the $i$-th operation is applied. Let $(\Delta \ell)_i$ denote the change of the number of leaf nodes in $T$ while the $i$-th operation is applied. Let $(\Delta m)_i$ denote the change of the number of dead leaf nodes in $T$ while the $i$-th operation is applied. The subscript $i$ may be removed when the context is clear.
We need to secure that $\Delta \ell + \Delta m \ge \Delta n /5$ holds for each of the following operations and the initial operation.

\begin{description}
\item [Operation 1:] If $T$ has a leaf node $x$ that has $d \ge 2$ neighbors outside $T$, then expand $T$ at $x$. In this case, $\Delta n \le d$, $\Delta \ell \ge d-1$, and $\Delta m \ge 0$. 

\item [Operation 2:] If every leaf node in $T$ has at most one neighbor outside $T$ and some node $x \notin T$ has $d \ge 2$ neighbors in $T$, then expand $T$ at one of $x$'s neighbors in $T$. In this case, $\Delta n \le 1$, $\Delta \ell = 0$, and $\Delta m = d-1$.

\item [Operation 3:] This operation is used only when the previous two operations do not apply. Let $x_0$ be some leaf in $T$ that has exactly one neighbor $x_1$ not in $T$. For each $i \ge 1$, if $x_{i}$ is defined and all neighbors of $x_i$ other than $x_{i-1}$ are outside $T$ and $x_i$ has degree two in $G$, then define $x_{i+1}$ to be the neighbor of $x_i$ other than $x_{i-1}$. Suppose that $x_i$ for $i \le k$ are defined and $x_{k+1}$ is not defined, then we expand $T$ at $x_i$ for each $i \le k$ in order. Though $k$ can be arbitrarily large, $\Delta n \le 2+\deg_{G}(x_k)$. If $x_{k+1}$ is not defined and $x_{k}$ has $d > 0$ neighbors other than $x_{k-1}$ in $T$ (thus $k \ge 2$ in this case  otherwise Operation 2 applies), then we discuss in subcases:

    \begin{description}
        \item [Subcase 1 ($\deg_{G}(x_k) = 1$):] It is impossible to have $\deg_{G}(x_k) = 1$ for this case.
        \item [Subcase 2 ($\deg_{G}(x_k) = 2$):] Then $\Delta \ell = 0$ and $\Delta m = 2$.
        \item [Subcase 3 ($\deg_{G}(x_k) \ge 3$):] Then $\Delta \ell = \deg_{G}(x_k)-d-2$ and $\Delta m \ge d$.
    \end{description}
If $x_{k+1}$ is not defined and $x_{k}$ has 0 neighbor other than $x_{k-1}$ in $T$, then $\deg_{G}(x_k)$ is either 1 or $\ge 3$. For $\deg_{G}(x_k) = 1$, $\Delta \ell = 0$ and $\Delta m = 1$. For $\deg_{G}(x_k) \ge 3$, $\Delta \ell = \deg_{G}(x_k)-2$ and $\Delta m \ge 0$. 
\end{description}

It is clear that one can expand $T$ to get a spanning tree of $G$ by a sequence of the above operations. Because all leaves are eventually dead, $\sum \Delta m = \sum \Delta \ell$. Consequently, $2\leaf(G) \geq 2 \sum \Delta \ell = \sum \Delta \ell + \Delta m \ge (\sum \Delta n)/5 = (V(G) - \inode{G})/5$, as desired.
}

Given \cref{lem:dead-leaf}, our goal is, for every constant $\varepsilon > 0$, find a sparse subgraph $H$ of the input graph $G$ so that:

\begin{enumerate}
\item The nodes incident to the edges in $T^* \setminus H$ can be \emph{dominated} by a small set $S$ of at most $\varepsilon(|V(G)|-\inode(G))$ nodes, i.e. either in $S$ or has at least one neighbor node in $S$ using the edges in $H$, where $T^*$ is any optimal MLST of $G$.
\item $H$ is connected.
\end{enumerate}

Because of the existence of the small dominating set $S$, one can obtain a forest $F$ from $T^* \cap H$ by adding some edges in $H$ so that the number of leaves in $F$ is no less than that in $T^*$ by $|S|$ and the number of connected components in $F$ is no more than that in $T^*$ by $|S|$. Since $H$ is connected, one can further obtain a spanning tree $T$ from $F$ by adding at most $|S|$ edges in $H$, so the number of leaves in $T$ is no less than that in $F$ by $2|S|$. Pick an $H$ associated with a sufficiently small $\varepsilon$, by \cref{eqn:leafy} $H$ is an $\varepsilon$-MLST sparsifier. A formal proof is given below.

\begin{theorem} \label{thm:mlst-sparsifier}
For every integer $k \ge 186$, every connected simple undirected graph $G$ has 
$$
    \leaf(S_k(G) \cup T(G)) \ge \left(1-30\left(\frac{1+\ln (k+1)}{k+1}\right)\right) \leaf(G).
$$ 
\end{theorem}
\metaproof{\cref{thm:mlst-sparsifier}}{
Let $T^*$ be a spanning tree of $G$ that has $\leaf(G)$ leaves. Let $k$ be some fixed integer at least $3$ and let $H = S_k(G) \cup T(G)$. Let $L = \{ x \in V(G) \colon x \mbox{ is incident to some } e \in T^* \setminus H \}$. Note that every node $x \in L$ has $\deg_{G}(x) > k$, so $x$ and all neighbors of $x$ are not ignorable nodes in $G$. 

First, we show that $L$ can be dominated by a small set $S$ of size at most $\varepsilon(|V(G)|-\inode(G))$ using some edges in $H$. We obtain $S$ from two parts, $S_1$ and $S_2$. $S_1$ is a random node subset sampled from the non-ignorable nodes in $G$, in which each node is included in $S_1$ with probability $p$ independently, for some $p \in (0, 1)$ to be determined later. Thus, $E[|S_1|] = p(|V(G)|-\inode(G))$. Since every node $x \in L$ is adjacent only to the non-ignorable nodes in $G$, the probability that $x \in L$ is not dominated by any node in $S_1$ is 
$$
    \Pr[x \mbox{ is not dominated}] = (1-p)^{1+\deg_{H}(x)} \le (1-p)^{k+1}.
$$
Let $S_2$ be the set of nodes in $L$ that are not dominated by any node in $S_1$ using the edges in $H$. Thus, 
$$
    E[|S|] = E[|S_1|+|S_2|] \le \left(p+(1-p)^{k+1}\right)(|V(G)|-\inode(G)).
$$

Then, we obtain a forest $F$ from $T^* \cap H$ by adding some edges in $H$ as follows. Initially, $F = T^* \cap H$.

\begin{description}
    \item [Operation 1:] For each $x \in L$, if $x$ is an isolated node in $T^* \cap H$ and $x \notin S$, then add an edge $e$ from $x$ to some node in $S$ to $F$. Such an edge $e$ must exist because $S$ dominates $L$.
    \item [Operation 2:] For each $x \in L$, if $x$ is not an isolated node in $T^* \cap H$ and the connected component that contains $x$ has an empty intersection with $S$, then add an edge $e$ from $x$ to some node in $S$ to $F$. Again, such an edge $e$ must exist because $S$ dominates $L$.
\end{description}

For each leaf $\ell \in T^*$, if $\deg_{G}(\ell) \le k$, then $\ell$ is a leaf in $T^* \cap H$ (also in $F$ unless $\ell \in S$); otherwise $\deg_{G}(\ell) > k$, if $\ell$ is not a leaf in $T^* \cap H$, then $\ell$ must be an isolated node in $T^* \cap H$, and by Operation 1 $\ell$ is connected to some node in $S$ unless $\ell \in S$. Hence, except those in $S$, every $\ell$ is a leaf node in $F$, so the number of leaves in $F$ is no less than that in $T^*$ by $|S|$. By Operation 2, the number of connected component is at most $|S|$.

Lastly, since $H$ is connected, one can obtain a spanning tree $T$ from $F$ by connecting the components in $F$ by some edges in $H$. Thus, the number of leaves in $T$ is no less than that in $T^*$ by $3|S|$. To obtain an $\varepsilon$-MLST sparsifier, by \cref{lem:dead-leaf} we need:
$$
    \frac{3|S|}{\frac{1}{10}(|V(G)|-\inode(G))} \le 30 \left(p+(1-p)^{k+1}\right) \le 30\left(p + e^{-p(k+1)} \right) \le \varepsilon 
$$
Setting $p = (\ln (k+1))/(k+1)$ gives the desired bound, and the leading constant is positive for $k \ge 186$.
}

    To find such a subgraph $H$, fetching a spanning tree of the input graph $G$ and grabbing $k$ edges for each node in $G$ suffices. Thus, we get a single-pass $\tilde{O}(n)$-space algorithm for the insertion-only model. As for the turnstile model, we use $\tilde{O}(k)$ $\ell_0$-samplers~\cite{Jowhari11} for each node in $G$ to fetch at least $k$ neighbors of $x$ w.h.p., and fetch a spanning tree by appealing to the single-pass $\tilde{O}(n)$-space algorithm for spanning trees in dynamic streams~\cite{AhnGM12}.
    This gives a proof of \cref{thm:main-mlst}.

\paragraph{Applications.} In~\cite{GandhiHKPS18}, Gandhi et al. show a connection between the maximum-leaf spanning trees and connected maximum cut. Their results imply that, for any unweighted regular graph $G$, the connected maximum cut can be found by the following two steps:

\begin{description}
    \item [Step 1:] Find a spanning tree $T$ whose $\leaf(T) \ge (1/2-\varepsilon)\leaf(G)$ for some constant $\varepsilon > 0$.
    \item [Step 2:] Randomly partition the leaves in $T$ into two parts $L$ and $R$ so that each leaf is included in $L$ with probability $1/2$ independently.  
\end{description}

Then, outputting $L$ and $V(G) \setminus L$ yields an $8+\varepsilon$-approximation for connected maximum cut. Step 1 is the bottleneck and can be implemented by combining our $\varepsilon$-MLST sparsifier (\cref{thm:main-mlst}) with the 2-approximation algorithm for MLST due to Solis{-}Oba, Bonsma, and Lowski~\cite{Solis-Oba17}. This gives  \cref{cor:cmcut}.

\begin{corollary}\label{cor:cmcut}
In the turnstile model, for every constant $\varepsilon > 0$, there exists a randomized algorithm that can approximate the connected maximum cut for $n$-node unweighted regular graphs to within a factor of $8+\varepsilon$ with probability $1-1/n^{\Omega(1)}$ in a single pass using $\tilde{O}(f(\varepsilon) n)$ space. 
\end{corollary}
\section{Breadth-First Search Trees} \label{sect.bfs}

A  BFS tree of an $n$-node connected simple undirected graph can be constructed in $O(n)$ passes using $\tilde{O}(n)$ space by simulating the standard BFS algorithm layer by layer. By storing the entire graph, a BFS tree can be computed in a single pass using $O(n^2)$ space. In Section~\ref{sect:bfs-1}, we show that it is possible to have a smooth tradeoff between pass complexity and space usage. In Section~\ref{sect:bfs-2}, we prove  \cref{thm:main-bfs}, which shows that the above tradeoff can be improved when randomness is allowed, even in the turnstile model. Then, in \cref{sect:bfs-3}, we show that multiple BFS trees, each starting from a distinct source node, can be computed more efficiently 
in a batch than individually. Lastly, we demonstrate an  application to diameter approximation in \cref{sect:bfs-4}.

In the BFS problem, we are given an $n$-node connected simple undirected graph $G = (V, E)$ and a distinguished node $s$, and it suffices to compute the distance $\dist(s, v)$ for each node $v \in V \setminus \{s\}$.  To infer a BFS tree from the distance information $\{ \dist(s, v)\colon v \in V\}$, it suffices to assign a parent to each node $v \in V \setminus \{s\}$ the smallest-identifier node from the set $\{u \in N(v) \colon \dist(s, u) = \dist(s, v) - 1\}$ where $N(v)$ is the set of $v$'s neighbors. This can be done with one additional pass using $\tilde{O}(n)$ space in the insertion-only model. In the turnstile model, for $p$-pass streaming algorithms with $p > \log n$, this can be done with $O(\log n/\log \log n)$ additional passes w.h.p. using $O(\log n)$ $\ell_0$-samplers~\cite{Jowhari11} for each node $v \in V \setminus \{s\}$, and this costs $\tilde{O}(n)$ space. For $p \le \log n$, the space bound is $\tilde{O}(n^2)$ and one can use $\tilde{O}(n)$ $\ell_0$-samplers for each node, so this step can be done in one additional pass. Hence in the subsequent discussion we focus on computing the distance from $s$ to each node $v \in V \setminus \{s\}$.

\subsection{A Simple Deterministic Algorithm \label{sect:bfs-1} }

We present a simple deterministic $p$-pass $\tilde{O}(n^2/p)$-space  algorithm in the insertion-only model by an observation that every root-to-leaf path in a BFS tree cannot visit too many high-degree nodes (\cref{lem:degsum}). Then, one can simulate the standard BFS algorithm efficiently layer-by-layer over high-degree nodes (\cref{thm:strbf}).

\begin{lemma}\label{lem:degsum}
Let $P$ be a root-to-leaf path in some BFS tree of an $n$-node connected simple undirected graph $G$. Then 
$$
    \sum_{x \in P} \deg_{G}(x) \le 3n = O(n)
$$
where $\deg_{G}(x)$ denotes the degree of $x$ in $G$. 
\end{lemma}
\metaproof{\cref{lem:degsum}}{
Suppose $P = x_1 x_2 \cdots x_k$ comprises $k$ nodes. Observe that if $x_i$ and $x_j$ have $i \equiv j \pmod 3$, then $x_i$ and $x_j$ cannot share any neighbor node; otherwise $P$ can be shorten, a contradiction. Thus, for each $c \in \{0, 1, 2\}$ the total contribution of all $x_i$'s whose $i \equiv c \pmod 3$ to $\sum_{x_i \in P} \deg_{G}(x_i)$ is $O(n)$. Summing over all possible $c$ gives the bound. 
}

We note that Lemma~\ref{lem:degsum} is near-optimal.
To see why, let $H = (V, E)$ where $V$ is the union of disjoint sets $V_0, V_1, \ldots, V_k$ and $E = \{(x, y) : x \in V_i \mbox{ and } y \in V_{j} \mbox{ for any } i, j \mbox{ that }|i-j| \le 1\}$. By setting $k = \lceil(n-1)/t\rceil$ for some parameter $t$, $|V_0| = 1$, $|V_i| = t$ for every $i \in [1, k-1]$, and $1 \leq |V_k| \leq t$, 
any BFS tree rooted at the node in $V_0$ has a root-to-leaf path $Q$ of length $k$, and each node in $Q \cap (V_2 \cup V_3 \cup \ldots \cup V_{k-2})$ has degree $3t - 1$. Pick any $t$ such that $k = \omega(1)$ and $t = \omega(1)$. We have $\sum_{x \in Q} \deg_{H}(x) = (3-o(1))n$.

\begin{theorem}\label{thm:strbf}
Given an $n$-node connected simple undirected graph $G$ with a distinguished node $s$, a BFS tree rooted at $s$ can be found deterministically in $p$ passes using $\tilde{O}(n^2/p)$ space for every $p \in [1, n]$ in the insertion-only model.
\end{theorem}
\metaproof{\cref{thm:strbf}}{
 Given a parameter $k$, our algorithm goes as follows. In the first pass, keep arbitrary $n/k$ neighbors for each node $v \in G$ in memory and then use the in-memory edges to update the distance $\dist(s, v)$ for each $v \in G$ by any single-source shortest path algorithm. The set of the in-memory edges is an invariant after the first pass. Hence, the memory usage is $\tilde{O}(n^2/k)$.
 Then, in each of the subsequent passes, processing the edges $(u, v)$ in the stream one by one, without keeping them in memory after the processing, if $\dist(s, u) + 1 < \dist(s, v)$ (resp. if $\dist(s, v) + 1 < \dist(s, u)$), then update $\dist(s, v)$ (resp. $\dist(s, u)$). After the edges in the stream are all processed, use the in-memory edges to update the distance $\dist(s, v)$ for each $v \in G$ again by any single-source shortest path algorithm but with initial distances. Our algorithm repeats until no distance has been updated in a single pass. 
 
 Observe a root-to-leaf path $P = s z_1 z_2 \cdots z_t$ in some BFS tree rooted at $s$. Suppose $P$ contains exactly $\ell$ edges that appears only on tape, let them be  $(z_{x_1}, z_{y_1}), \ldots,  (z_{x_\ell}, z_{y_\ell})$ where $1 \le x_i < y_i \le x_{i+1} < y_{i+1} \le t$ for every $i \in [1, \ell-1]$. Let $\pred_P(z_i)$ be the predecessor of $z_i$ on $P$ that is closest to $z_i$ among nodes in $\{s\} \cup \{z_{y_j} : y_j < i\}$. By the definition of the above construction, it is assured that $\deg(z_{x_i}) \ge n/k$ for each $i \in [1, \ell]$. Thus by~\cref{lem:degsum}, $\ell = O(k)$. Then we appeal to the argument used for the analysis of Bellman-Ford algorithm~\cite{Ford56, Bellman58}. For every $i \in [1, t]$, if $i \notin \{y_1, y_2, \ldots, y_\ell\}$, $\dist(s, z_i)$ attains the minimum possible value at the same pass when $\dist(s, \pred_P(z_i))$ attains; otherwise $i = y_j$ for some $j \in [1, \ell]$, $\dist(s, y_j)$ attains the minimum possible value at most one pass after $\dist(s, x_j)$ attains. Hence, $O(k)$ passes suffices to compute $\dist(s, z_i)$ for all $i \in [1, t]$ and this argument applies to all root-to-leaf paths.  Setting $k = p$ yields the desired bound.
}

\subsection{A More Efficient Randomized Algorithm \label{sect:bfs-2}}

In this section, we prove 
Theorem~\ref{thm:main-bfs}. Our BFS algorithm is based on the following generic framework, which has been applied to finding shortest paths in the parallel and the distributed settings~\cite{Elkin17,GhaffariL17,HuangNS17,UllmanY91}. Sample a set $U$ of approximately $k$ distinguished nodes such that each node $v \neq s$  joins $U$ independently with probability $k / n$, and $s \in U$ with probability 1. By a Chernoff bound, $|U|  = \tilde{\Theta}(k)$ with high probability. We will grow a local BFS tree of radius $\tilde{O}(n/k)$ from each node in $U$, and then we will construct the final BFS tree by combining them. We will rely on the following lemma, which first appeared in~\cite{UllmanY91}.

\begin{lemma}[\cite{UllmanY91}]\label{lem:center}
Let $s$ be a specified source node.
Let $U$ be a subset of nodes such that each node $v \neq s$ joins $U$ with probability $k / n$, and $s$ joins $U$ with probability 1.
For any given parameter $C \geq 1$,  the following holds with probability $1 - n^{-\Omega(C)}$.
For each node $t \neq s$, there is an $s$-$t$ shortest path $P_{s,t}$ such that each of its $C (n \log n)/k$-node subpath $P'$ satisfies $P' \cap U \neq \emptyset$.
\end{lemma}

For notational simplicity, in subsequent discussion we write $h = C (n \log n)/k - 1 = \tilde{O}(n/k)$.
Lemma~\ref{lem:center} shows that for each node $t \in V \setminus \{s\}$, 
\begin{equation}\label{eqn:dist1}
\dist(s,t) = \min_{u \in U \cap N^h(t)} \dist(s,u) + \dist(u,t)
\end{equation}
with probability $1-n^{-\Omega(C)}$ where $N^h(v) = \{u \colon \dist(u, v) \le h\}$.

To see this, consider the $s$-$t$ shortest path $P_{s,t}$ specified in Lemma~\ref{lem:center}. 
If the number of nodes in $P_{s,t}$ is less than $h$, then the above claim holds because $s \in U \cap N^h(t)$.
Otherwise, Lemma~\ref{lem:center} guarantees that there is a node $u \in P_{s,t} \cap U \cap N^h(t)$ with  probability  $1 - n^{-\Omega(C)}$.
Using \cref{eqn:dist1}, 
a BFS tree can be computed using the following steps.
\begin{enumerate}
    \item Compute $\dist(u,v)$ for each $u \in U$ and $v \in U \cap N^h(u)$. Using this information, we can infer $\dist(s, u)$ for each $u \in U$.
    \item Compute $\dist(s,t)$ for each $t \in V \setminus \{s\}$ by the formula $\dist(s,t) = \min_{u \in U \cap N^h(t)} \dist(s,u) + \dist(u,t)$.
\end{enumerate}

In what follows, we show how to implement the above two steps in the streaming model, using  $\tilde{O}(n + k^2)$ space and $\tilde{O}(n/k)$ passes. By a change of parameter $p = \tilde{O}(n/k)$, we obtain Theorem~\ref{thm:main-bfs}.

\paragraph{Step 1.}
To compute $\dist(u,v)$ for each $u \in U$ and $v \in U \cap N^h(u)$, we let each $u \in U$ initiate a radius-$h$ local BFS rooted at $u$. A straightforward implementation of this approach in the streaming model costs $h = \tilde{O}(n / k)$ passes and $O(n \cdot |U|) = \tilde{O}(nk)$ space, since we need to maintain $|U|$  search trees simultaneously.

We show that the space requirement can be improved to $\tilde{O}(n + k^2)$. Since we only need to learn the distances between nodes in $U$, we are allowed to forget distance information associated with nodes $v \notin U$   when it is no longer needed. Specifically,
suppose we start the BFS computation rooted at $u \in U$ at the $\tau_u$th pass, where  $\tau_u$ is some number to be determined. For each $0 \leq i \leq h-1$, the induction hypothesis specifies that at the beginning of the $(\tau_u + i)$th pass, all nodes in $L_i(u) = \{ v \in V \colon \dist(u,v) = i\}$ have learned that $\dist(u,v) = i$. During the $(\tau_u + i)$th pass, for each node $v \in V$ with  $\dist(u,v) > i$, we check if $v$ has a neighbor in $L_i(u)$. If so, then we learn that $\dist(u,v) = i+1$.

In the above BFS algorithm,  if $\dist(u,v) = i$ for some $0 \leq i \leq h-1$, then we learn the fact  that $\dist(u,v) = i$ during the $(\tau_u + i-1)$th pass. Observe that such information is only needed during the next two passes. After the end of the $(\tau_u + i+1)$th pass, for each $v \in V$ with $\dist(u,v) = i$, we are allowed to forget that $\dist(u,v) = i$. That is, $v$ only needs to participate in the BFS computation rooted at $u$ during these three passes $\{ \tau_u + i-1, \ \tau_u + i, \ \tau_u + i + 1\}$.

For each $u \in U$, we assign the starting time $\tau_u$ independently and uniformly at random from $\{1, 2, \ldots, h\}$.  Lemma~\ref{lem:aux} shows that for each node $v \in V$ and for each pass $1 \leq t \leq 2h-1$, the number of BFS computations that involve $v$ is $\tilde{O}(1)$. The idea of using random starting time to schedule multiple algorithms to minimize congestion can be traced back from~\cite{LeightonMR94}. Note that $\tau_u + \dist(u,v) -1 \leq t \leq \tau_u + \dist(u,v) + 1$ is the criterion for $v$ to participate in the BFS rooted at $u$ during the $t$th pass.

\begin{lemma}\label{lem:aux}
For each node $v$, and for each integer $1 \leq t \leq 2h-1$, with high probability, the number of nodes $u \in U$ such that $\tau_u + \dist(u,v) -1 \leq t \leq \tau_u + \dist(u,v) + 1$ is at most $O(\max\{\log n, |U| / h\})$.
\end{lemma}
\metaproof{\cref{lem:aux}}{
Given two nodes $u \in U$ and $v \in V$, and a fixed number $t$,
the probability that $\tau_u + \dist(u,v) -1 \leq t \leq \tau_u + \dist(u,v) + 1$ is at most $3 / h$.
Let $X$ be the total number of $u \in U$ such that $\tau_u + \dist(u,v) -1 \leq t \leq \tau_u + \dist(u,v) + 1$.
The expected value of $X$ can be upper bounded by $\mu =  |U| \cdot (3/h)$.
By a Chernoff bound, with high probability, $X = O(\max\{\log n, |U| / h\})$.
}

Recall that $|U| = \tilde{O}(k)$ with high probability, and $h = \tilde{O}(n/k)$.
By  Lemma~\ref{lem:aux}, we only need $\ceil{k^2 / n} \cdot \tilde{O}(1)$ space per each $v \in V$ to do the radius-$h$ BFS computation from all nodes $u \in U$. That is, the space complexity is $\tilde{O}(n + k^2)$. To store the  distance information   $\dist(u,v)$ for each $u \in U$ and $v \in U \cap N^h(u)$, we need $\tilde{O}(k^2)$ space. Thus, the algorithm for Step~1 costs $\tilde{O}(n + k^2)$ space. The number of passes is $2h-1 = \tilde{O}(k)$.

In the insertion-only model, the implementation is straightforward. In the turnstile model,  care has to be taken when implementing the above algorithm. 
We write $x = O(\max\{\log n, |U| / h\})$ to be the high probability upper bound on the number of BFS computation that a node participates in a single pass. We write $y = O(x \log n)$.
Let $U_1, U_2, \ldots, U_y$ be random subsets of $U$ such that  each $u \in U$ joins each $U_j$ with probability $1/x$, independently.
Consider a node $v \in V$ and consider the $r$th pass. Let $S$ be the subset of $U$ such that $u \in S$ if  $r = \tau_u +  \dist(u,v) - 1$, i.e., the BFS computation rooted at $u$ hits $v$ during the $r$th pass. We know that with high probability $|S| \leq x$.
By our choice of $U_1, U_2, \ldots, U_y$, we can infer that  with high probability for each $u \in S$ there is at least one index $j$ such that $S \cap U_j = \{u\}$. 

To implement the $r$th pass in the turnstile model, each node $v \in V$ virtually maintains $y$ edge set $Z_1, Z_2, \ldots, Z_y$. For each insertion (resp.,  deletion) of an edge $e = \{w,v\}$ satisfying $r = \tau_u +  \dist(u,w) - 2$ for some $u \in U_j$, we add (resp., remove) the edge from the set $Z_j$. After processing the entire data stream, we take one edge out of each edge set  $Z_1, Z_2, \ldots, Z_y$. In view of the above discussion, it suffices to only consider these edges when we grow the BFS trees. This can be implemented using $y$ $\ell_0$-samplers per each node $v \in V$, and the space complexity is still $\tilde{O}(ny) = \tilde{O}(n + k^2)$.

\paragraph{Step 2.}
At this moment we have computed $\dist(s,u)$ for each $u \in U$.
 Now we need to compute $\dist(s,t)$ for each $t \in V \setminus \{s\}$ by the formula $\dist(s,t) = \min_{u \in U \cap N^h(t)} \dist(s,u) + \dist(u,t)$.
 
In the insertion-only model, this task can be solved using $h$ iterations of Bellman-Ford steps. Initially, $d_0(v) = \dist(s,v)$ for each $v \in U$, and $d_0(v) = \infty$ for each $v \in V \setminus U$.
During the $i$th pass, we do the update $d_i(v) \leftarrow \min\{ d_{i-1}(v), \ 1+\min_{u \in N(v)} d_{i-1}(u)\}$. By \cref{eqn:dist1}, we can infer that $d_h(t) = \dist(s,t)$ for each $t \in V$. A straightforward implementation of this procedure costs  $\tilde{O}(n)$ space and $h = \tilde{O}(n/k)$ passes.

In the turnstile model, we can solve this task by growing a radius-$h$ BFS tree rooted at $u$, for each $u \in U$, as in Step~1. During the process, each node $v \in V$ maintains a variable $d(v)$ which serves as the estimate of $\dist(s,v)$. Initially, $d(v) \leftarrow \infty$. When the partial BFS tree rooted at $u \in U$ hits $v$, we update  $d(v)$ to be the minimum of the current value of $d(v)$ and $\dist(s,u) + \dist(u,v)$. At the end of the process, we have $d(v) = \min_{u \in U \cap N^h(t)} \dist(s,u) + \dist(u,v) = \dist(s,v)$ for each $v \in V$.
This costs  $\tilde{O}(n + k^2)$ space and $\tilde{O}(n/k)$ passes in view of the analysis of Step~1.

\subsection{Extensions}\label{sect:bfs-3}

In this section we consider the problem of solving $c$ instances of BFS simultaneously for some $c \le n$  and a simpler problem of computing the pairwise distance between the $c$ given nodes. Both of these problems can be solved via a black box application of \cref{thm:main-bfs}. In this section we show that it is possible to obtain better upper bounds.

\begin{theorem}\label{thm:c-distance}
Given an $n$-node undirected graph $G$, for any given parameters $1 \leq c \leq k \leq n$, the pairwise distances between all pairs of nodes in a given set of $c$ nodes in $G$ can be computed with probability $1-1/n^{\Omega(1)}$ using $\tilde{O}(n/k)$ passes and $\tilde{O}(n + k^2)$ space in the turnstile model.
\end{theorem}
\metaproof{\cref{thm:c-distance}}{
Let $S$ be the input node set of size $c$. Consider the modified Step~1 of our algorithm where each $s \in S$ is included in $U$ with probability 1. Since $|S| = c \leq k$, we still have  $|U| = \tilde{O}(k)$ with high probability. Recall that Step~1 of our algorithm calculates $\dist(u,v)$ for each $u \in U$ and $v \in U \cap N^h(u)$ in $\tilde{O}(n + k^2)$ space and $\tilde{O}(n/k)$ passes. Applying \cref{eqn:dist1} for each $s \in U$,  we obtain the pairwise distances between all pairs of nodes in $U$, which includes $S$ as a subset. There is no need to do Step~2.
}

For example, if $c = n^{1/2}$, then Theorem~\ref{thm:c-distance} implies that we can compute  the pairwise distances between all pairs of nodes in a given set of $c$ nodes   in  $\tilde{O}(n)$ space and $\tilde{O}(n^{1/2})$ passes.

\begin{theorem}\label{thm:c-bfs}
Given an $n$-node undirected graph $G$, for any given parameters $1 \leq c \leq k \leq n$, one can solve $c$ instances of BFS with probability $1-1/n^{\Omega(1)}$ using $\tilde{O}(n/k)$ passes and $\tilde{O}(cn + k^2)$ space in the turnstile model.
\end{theorem}
\metaproof{\cref{thm:c-bfs}}{
Let $S$ be the node set of size $c$ corresponding to the roots of the $c$ BFS instances. Consider the following modifications to our BFS algorithm.

Same as the proof of Theorem~\ref{thm:c-distance}, in Step~1, include each $s \in S$ in $U$ with probability 1. The modified Step~1 still takes $\tilde{O}(n + k^2)$ space and $\tilde{O}(n/k)$ passes, and it outputs the pairwise distances between all pairs of nodes in $U$.

Now consider Step~2. In the insertion-only model, remember that a BFS tree rooted at a node $s \in S$ can be constructed in  ${O}(n)$ space and $h = \tilde{O}(n/k)$ passes using $h$ iterations of Bellman-Ford steps. The cost of  constructing all $c$ BFS trees is then  ${O}(cn)$ space and $\tilde{O}(n/k)$ passes.

In the turnstile model, we can also use the strategy of   growing a radius-$h$ BFS tree rooted at $u$, for each $u \in U$. During the process, each node $v \in V$ maintains $c$ variables  serving as the estimates of $\dist(s,v)$, for all $s \in S$.  The complexity of growing radius-$h$ BFS trees is still $\tilde{O}(n + k^2)$ space and $\tilde{O}(n/k)$ passes. The extra space cost for maintaining these $cn$ variables is $O(cn)$.
}

For example, if $c = n^{1/3}$, then Theorem~\ref{thm:c-bfs} implies that we can solve $c$ instances of BFS  in  $\tilde{O}(n^{4/3})$ space and $\tilde{O}(n^{1/3})$ passes. Note that the space complexity of $\tilde{O}(n^{4/3})$ is necessary to output $c = n^{1/3}$ BFS trees. 

\cref{thm:c-bfs} immediately gives the following corollary. 

\begin{corollary}\label{cor:Steiner}
Given an $n$-node connected undirected graph $G$ with unweighted edges and a $c$-node subset $S$ of $G$, for any given parameters $1 \le c \le k \le n$, finding a Steiner tree in $G$ that spans $S$ can be approximated to within a factor of $2$ with probability $1-1/n^{\Omega(1)}$ using $\tilde{O}(n/k)$ passes and $\tilde{O}(cn+k^2)$ space in the turnstile model.
\end{corollary}

Note that if we do not need to construct a Steiner tree, and only
 need to  approximate the size of an optimal Steiner tree, then \cref{thm:c-distance} can be used in place of  \cref{thm:c-bfs}.

\subsection{Diameter Approximation}\label{sect:bfs-4}

It is well-known that the maximum distance label in a BFS tree gives a $2$-approximation of diameter. We show that it is possible to improve the approximation ratio to nearly $1.5$ without sacrificing the space and pass complexities.

Roditty and Williams~\cite{RodittyW13} showed that a nearly $1.5$-approximation of diameter can be computed with high probability as follows. 
\begin{enumerate}
    \item Let $S_1$ be a node set chosen by including each node $v \in V$ to  $S_1$ with probability $p =  (\log n) /
    \sqrt{n}$ independently. Perform a BFS from each node $v \in S_1$.
    \item Let $v^\star$ be a node chosen to maximize $\dist(v^\star, S_1)$. Break the tie arbitrarily.  Perform a BFS from $v^\star$.
    \item Let $S_2$ be the node set consisting of the $\sqrt{n}$ nodes closest to $v^\star$, where ties are broken arbitrarily.  Perform a BFS from each node $v \in S_2$.
\end{enumerate}
Let $D^\ast$ be the maximum distance label ever computed during the BFS computations in the above procedure. Roditty and Williams~\cite{RodittyW13} proved that $D^\ast$ satisfies that $\lfloor 2D/3 \rfloor  \leq   D^\ast   \leq  D$, where $D$ is the diameter of $G$.

The algorithm of Roditty and Williams~\cite{RodittyW13} can be implemented in the streaming model by applying Theorem~\ref{thm:c-bfs} with $c = \tilde{O}(\sqrt{n})$, but we can do better.
Note that when we perform  BFS from the nodes in $S_1$ and $S_2$, it is not necessary to store the entire BFS trees. For example, in order to select $v^\ast$, we only need to let each node $v$ know $\dist(v, S_1)$, which is the maximum distance label of $v$ in all BFS trees computed in Step~1.
Therefore, the $O(cn)$ term in the space complexity of  Theorem~\ref{thm:c-bfs} can be improved to $O(n)$. That is, the space and pass complexities are the same as the cost for computing a \emph{single} BFS tree using \cref{thm:main-bfs}.  We conclude the following theorem.

\begin{theorem}\label{thm:diam-apx}
Given an $n$-node connected undirected graph $G$, a diameter approximation $D^\ast$ satisfying  $\lfloor 2D/3 \rfloor  \leq   D^\ast   \leq  D$, where $D$ is the diameter of $G$,  can be computed with probability $1-1/n^{\Omega(1)}$ in $p$ passes using $\tilde{O}((n/p)^2)$ space, for each $1 \leq p \leq \tilde{O}(\sqrt{n})$ in the turnstile model.
\end{theorem}
\section{Depth-First Search}\label{sect.dfs}
A straightforward implementation of the naive DFS algorithm in the streaming model costs either $n-1$ passes with $\tilde{O}(n)$ space or a single pass with $O(n^2)$ space.
Khan and Mehta~\cite{khan2019depth} recently showed that it is possible to obtain a smooth tradeoff between the two extremes. Specifically, they designed an algorithm that requires at most
$\ceil{n/k}$ passes  using $\tilde{O}(nk)$ space, where $k$ is any positive integer. Furthermore, for the case the height $h$ of the computed DFS tree is small, they further decrease the  number of passes to $\ceil{h/k}$.  In Section~\ref{sect-dfs1}, we will provide a very simple alternative proof of this result, via sparse certificates for $k$-node-connectivity.

In the worst case, the ``space $\times$ number of passes'' of the algorithms of Khan and Mehta~\cite{khan2019depth} is still $\tilde{O}(n^2)$.
In Sections~\ref{sect-dfs2} and~\ref{sect-dfs3}, we will show that it is possible to improve  this upper bound asymptotically when the number of passes $p$ is super-constant. Specifically, for any parameters $1 \leq s \leq k \leq n$, we obtain the following DFS algorithms.
\begin{itemize}
    \item A deterministic algorithm using $\tilde{O}((n/k) + (k/s))$ passes and $\tilde{O}(ns)$ space in the insertion-only model. After balancing the parameters, the space complexity is $\tilde{O}(n^2  / p^2)$ for $p$-pass algorithms, for each $1 \leq p \leq \tilde{O}(\sqrt{n})$.
    \item A randomized algorithm using $\tilde{O}((n/k) + (k/s))$ passes and $\tilde{O}(n s^2)$ space in the turnstile model. After balancing the parameters, the space complexity is $\tilde{O}(n^3 / p^4)$ for $p$-pass algorithms, for each $1 \leq p \leq \tilde{O}(\sqrt{n})$.
\end{itemize}

\subsection{A Simple DFS Algorithm \label{sect-dfs1}}

In this section, we present a simple alternative proof of the result of Khan and Mehta~\cite{khan2019depth} that a DFS tree can be constructed in 
$\ceil{h/k}$ passes  using $\tilde{O}(nk)$ space, for any given parameter $k$, where $h$ is the height of the computed DFS tree.

\paragraph{Sparse Certificate for $s$-Node-Connectivity.}
A \emph{strong $s$-VC certificate} of a graph $H$ is its subgraph $K$ such that for any supergraph $G$ of $H$, for every pair of nodes $s^\ast, t^\ast \in G$, if  they are $c$-node-connected in $G$, then they are $c'$-node-connected for some $c' \ge \min\{s, c\}$ in the graph obtained from $G$ by replacing its subgraph $H$ with $K$. 
A \emph{sparse}   strong $s$-VC certificate of the graph $G$ is exactly what we need here.  Eppstein, Galil, Italiano, and Nissenzweig~\cite{EppsteinGIN97} showed that such a sparse subgraph of $O(ns)$ edges can be computed in a single pass with $\tilde{O}(ns)$ space \emph{deterministically} in the insertion-only model. In the turnstile model,  Guha,  McGregor, and Tench~\cite{GuhaMT15} showed that such a sparse subgraph of $\tilde{O}(ns^2)$ edges can be computed with high probability in a single pass using $\tilde{O}(ns^2)$ space. This result can be inferred from Theorem~8 of~\cite{GuhaMT15} with $\epsilon  = \Theta(1/s)$. In~\cite{GuhaMT15} the analysis only considers the case $G = H$, but it is straightforward to extend the analysis to incorporate any supergraph $G$ of $H$.

If the subgraph $K$ of the graph $H$ satisfies the above requirement for the special case of $G = H$, then $K$ is said to be a \emph{$s$-VC certificate} of $H$. Our simple DFS algorithm relies on this tool.

\begin{lemma}\label{lem:simple-dfs-aux}
Suppose $K$ is a $(k+1)$-VC certificate of $H$. Let $T$ be any DFS tree of $K$. Consider any two nodes $u$ and $v$ such that the least common ancestor $w$ of $u$ and $v$ are within the top $k$ layers of $T$. If $w \neq u$ and $w \neq v$, then $u$ and $v$ are not adjacent in $H$.
\end{lemma}
\metaproof{\cref{lem:simple-dfs-aux}}{
Suppose $u$, $v$, and $w$ violate the statement of the lemma. That is, $u$ and $v$ are adjacent in $H$. Since $T$ is a DFS tree, $u$ and $v$  are not adjacent in $K$, and each path connecting $u$ and $v$ must pass through a node that is a common ancestor of $u$ and $v$. 
Let $c_H$ (resp., $c_K$) be the $u$-$v$ node-connectivity in $H$ (resp., $K$).
The above discussion implies that $c_K \leq k$ and $c_H \geq c_K + 1$, contradicting the assumption that  $K$ is a $(k+1)$-VC certificate of $H$. 
}

\paragraph{Algorithm.} Using Lemma~\ref{lem:simple-dfs-aux}, we can construct a DFS tree of $G$ recursively as follows. Pick $K$ as a $(k+1)$-VC certificate of $G$. Compute a DFS tree $T$ of $K$. Let $T'$ be the tree induced by the  top $k+1$ layers of  of $T$. Let $v_1, v_2, \ldots, v_z$ be the leaves of $T'$. Denote $S_i$ as the set of descendants of $v_i$ in $T$, including $v_i$.
By Lemma~\ref{lem:simple-dfs-aux}, there exists no edge in $G$ that crosses two distinct sets $S_i$ and $S_j$. 
For each $1 \leq i \leq z$, we  recursively find a DFS tree $T_i$ of the subgraph of $G$ induced by $S_i$ rooted at $v_i$. By the above observation, we can obtain a valid DFS tree of $G$ by appending $T_1, T_2, \ldots, T_z$ to $T'$.

\paragraph{Analysis.} If the height of the final DFS tree is $h$, then the depth of the recursion is at most $\ceil{h/k}$. The cost for computing a $(k+1)$-VC certificate is $1$ pass and $\tilde{O}(nk)$  space, and the resulting subgraph $K$ has $O(nk)$ edges. Therefore, the  total number of passes is at most $\ceil{h/k}$, and the overall space complexity is $\tilde{O}(nk)$.

\subsection{Streaming Implementation of the  Algorithm of Aggarwal and Anderson \label{sect-dfs2}}

The bounds of Theorem~\ref{thm:main-dfs} are attained via an implementation of the parallel DFS algorithm of Aggarwal and Anderson~\cite{Aggarwal1988} in the streaming model, with the help of various tools, including the strong sparse certificates for $s$-node-connectivity described above. 

\paragraph{Overview.} At a high level, the DFS algorithm of Aggarwal and Anderson~\cite{Aggarwal1988} works as follows. Start with a maximal matching, and then merge these length-1 paths iteratively into a constant number of node-disjoint paths such that the number of nodes not in any path is at most $|V| / 2$. The algorithm then constructs the initial segment of the DFS tree from these paths. Each remaining connected component is solved recursively. The final DFS tree is formed by appending the DFS trees of recursive calls to the initial segment. 

The bottleneck of this DFS algorithm is a task called 
$\maximalpaths$ which is a variant of the maximal node-disjoint paths problem between a set of source nodes $S$ and a set of sink nodes $T$. In this variant, each member of $S$ is a path instead of a node.
Goldberg, Plotkin, and Vaidya~\cite{GPV93} gave a parallel algorithm for this problem. Their algorithm has two phases. 
For any given parameter $k$, they showed that after $k$ iterations of the algorithm of the first phase, the number of sources in $S$ that are still \emph{active} is at most $n / k$. These remaining active sources are processed one-by-one in the second phase.
Using this approach with $k = \sqrt{n}$,  $\maximalpaths$  can be solved in  the streaming model with $\tilde{O}(\sqrt{n})$ passes and $\tilde{O}(n)$ space. To further reduce the pass complexity, we apply the sparse certificates for $s$-node-connectivity of
Eppstein, Galil, Italiano, and Nissenzweig~\cite{EppsteinGIN97}
and Guha,  McGregor, and Tench~\cite{GuhaMT15}, which allow us to process the remaining active sources in batches.
In the insertion-only model, we obtain a deterministic $p$-pass algorithm with space complexity $\tilde{O}(n^2/p^2)$, for each  $1 \leq p \leq  \tilde{O}(\sqrt{n})$.
For the more challenging turnstile model, we obtain a randomized algorithm  with a somewhat worse space complexity of  $\tilde{O}(n^3/p^4)$.


\paragraph{The DFS Algorithm of Aggarwal and Anderson.} Specifically, the DFS algorithm of Aggarwal and Anderson~\cite{Aggarwal1988} is based on the following divide-and-conquer approach. The goal is to find a DFS tree of $G$ rooted at a given node $r$. To do so,  Aggarwal and Anderson~\cite{Aggarwal1988} devised an algorithm that finds a subtree $T$, called \emph{initial segment}, rooted at   $r$,  satisfying the following properties:
\begin{itemize}
    \item Each of the connected components $C_1, C_2, \ldots, C_z$ of $G \setminus T$ has at most $n/2$ nodes.
    \item The subtree $T$ can be extended to a DFS tree of $G$ as follows. For each connected component $C_i$, there is a unique node $v_i \in T$ of the largest depth in $T$ that is adjacent to nodes in $C_i$. Choose $r_i$ to be any node in $C_i$ adjacent to $v_i$. For each $1 \leq i \leq z$, append to $v_i$ any DFS tree of $C_i$ rooted at $r_i$.
\end{itemize}

It is clear that this gives a recursive algorithm with a logarithmic depth of recursion. In the insertion-only model, finding the portals $v_i$ and $r_i$ is straightforward and can be done in a single pass with $z = \tilde{O}(n)$ space, simultaneously for all $i = 1, \ldots, z$. In the turnstile model, we  employ a binary search on the depth of $v_i$ in $T$, and this costs $O(\log n)$ passes with $z = \tilde{O}(n)$ space. 

\paragraph{Constructing the Initial Segment.}
The initial segment $T$ is constructed in two steps. The first step is to find a set of node-disjoint paths  $Q$ of size at most 11, called \emph{separator}, such that each connected component of the subgraph induced by all nodes not in a path of $Q$ has at most $n/2$ nodes. 

The second step is to construct $T$ from $Q$ as follows. Initially, the subtree $T \leftarrow r$ consists of only the root node. While $Q$ is not empty, we extend the current subtree $T$ as follows. Find a path $p$ connecting a node $u$ in a path of $Q$ to a node $v$ in $T$ such that all intermediate nodes of $p$ are not in a path of $Q$ and are not in $T$. The path $p$ is chosen such that the depth of $v$ is the largest possible. Let $p'=(s, \ldots, u, \ldots, t) \in Q$ be the path that contains $u$. Extend the subtree $T$ by appending to $v$ the path $p = (v, \ldots, u)$ and the longer one the two subpaths $(s, \ldots, u)$ and $(u, \ldots, t)$ of $p'$. Then update $Q$ by removing from $p'$ the part that has been added to $T$. It is clear that $Q$ becomes empty after $O(\log n)$ iterations.

Implementation of the above procedure to the streaming model is also straightforward. We  do a binary search on the depth $d^\ast$ of $v$ to find  the path $p$. Specifically, for a parameter $d$, consider the subgraph $G_d$ induced by all nodes in $G$ except the ones in $T$ of depth greater than $d$. Compute any spanning forest $T_d$ of $G_d$. If all nodes in the paths of $Q$ are not reachable to all nodes in $T$ in the spanning forest $T_d$, then we know that $d < d^\ast$; otherwise $d \geq d^\ast$. After we have determined $d = d^\ast$, it suffices to pick $p$ as any minimal-length path connecting $T$ and $Q$ in the spanning forest $T_{d^\ast}$. The construction of a spanning forest can be done in a single pass with $\tilde{O}(n)$ space in the insertion-only model. For the turnstile model, we use the algorithm of Ahn, Guha, and McGregor~\cite{AhnGM12}, which also costs  $\tilde{O}(n)$ space and finishes in a single pass.

\paragraph{Constructing the Separator.} The algorithm for constructing $Q$ is as follows. At the beginning, $Q$ is initialized as any maximal matching. Obviously, each connected component induced by nodes not involved in $Q$ is a single node, but $|Q|$ can be as large as linear in $n$. The size of the set $Q$ can be decreased to at most $11$ by repeatedly applying the procedure $\reduce(Q)$ for $O(\log n)$ times.

If we are given a set of node-disjoint paths $Q$ such that $|Q| \geq 12$ and each connected component induced by nodes not involved in $Q$ has at most $n/2$ nodes, the  procedure $\reduce(Q)$ of~\cite{Aggarwal1988} is guaranteed to output a new set of node-disjoint paths $Q'$ such that $|Q'| \geq (11/12) |Q|$ and each connected component induced by nodes not involved in $Q'$ also has at most $n/2$ nodes.

Note that a maximal matching can be found via a greedy algorithm in a single pass with $\tilde{O}(n)$ space in the insertion-only model. In the turnstile model, a maximal matching can be found with high probability in $O(\log n)$ passes with  $\tilde{O}(n)$ space by implementing the parallel maximal matching algorithm of Israeli and Itai~\cite{II86} using $\ell_0$-samplers. 

\paragraph{Finding Node-Disjoint Paths.}
The detailed description of $\reduce(Q)$ is omitted.   All of  $\reduce(Q)$ can be implemented in the streaming model in $\tilde{O}(1)$ passes and $\tilde{O}(n)$ space, except the following task, called $\maximalpaths$~\cite{GPV93}. The input of $\maximalpaths$ consists of a set of source nodes $S \subseteq V$, a set of sink nodes $T \subseteq V$, and a set of node-disjoint directed paths $\Pin$ in $G$, where each source node $v \in S$ is the starting node of a path $P \in \Pin$.  The output of $\maximalpaths$ is a set of node-disjoint paths in $G$ such that each $P \in \Pout$ is of the form $P = s \circ P_1 \circ P_2 \circ t$ such that (i) $s \in S$, (ii) $t \in T$, (iii) $s \circ P_1$ is the prefix of some path in $\Pin$, and (iv) $P_2$ is a path that does not involve any nodes used in $\Pin$ and $T$. Note that $P_1$ and $P_2$ might be empty. The set $\Pout$ has to satisfy the following maximality constraint. For each node $v$ in a path of $\Pin$ but not in a path of $\Pout$, any path connecting $v$ to a sink node-intersects a path in $\Pout$. 

Note that in~\cite{Aggarwal1988} the sinks $T$ are node-disjoint paths, not individual nodes. Here each node in $T$ corresponds to the result of contracting each of these paths into a node.
Goldberg, Plotkin, and Vaidya~\cite{GPV93} showed that $\maximalpaths$ can be solved in two stages as follows.

\paragraph{First Stage.}
In the first stage, each node has three possible states: $\{\idle, \aactive, \dead\}$. Intuitively, the  $\dead$ nodes are the ones that will not be considered in subsequent steps of the algorithm. The set of \emph{active paths} $\Pactive$ is initialized as $\Pin$. All nodes in a path of  $\Pactive$ are $\aactive$. All remaining nodes are initially $\idle$.

In each iteration, the set of active paths $\Pactive$ are updated as follows. Let $H$ be the set of the last nodes in a path in  $\Pactive$. Let $H'$ be the set of $\idle$ nodes. Find a maximal matching $M$ on the bipartite graph induced by the two parts $H$ and $H'$. If a path $P \in \Pactive$ is incident to a matched edge $e=\{u,v\} \in M$, then $P$ is extended by appending $e =\{u,v\}$ to the last node $u$ of $P$, and the state of $v$ is updated to $\aactive$. Otherwise, the last node $u$ of $P  \in \Pactive$ is removed from $P$, and the state of  $u$  is  updated to $\dead$. 

A source is successfully connected to a sink when there is a path $P \in \Pactive$ that reaches a sink node. When this occurs, the entire path $P$ is removed from $\Pactive$ and is added to $\Pout$. All  nodes of $P$ are then $\dead$, as they should not be considered in subsequent steps.

The first stage terminates once $|\Pactive| < k$, where $k$ is a given parameter. 
Observe that the number of iterations can be upper bounded by $2n/k$, as  the number of nodes that change their states in an iteration is at least the number of active paths at the beginning of this iteration, and each node $v \in V$ can change its state at most twice.

Now consider the implementation in the streaming model. Recall that a maximal matching can be found deterministically in a single pass with $\tilde{O}(n)$ space in the insertion-only model, or in the turnstile model with high probability in $O(\log n)$ passes with  $\tilde{O}(n)$ space using the algorithm of Israeli and Itai~\cite{II86} via $\ell_0$-samplers. 
Hence the algorithm for the first stage can be implemented using  $\tilde{O}(n/k)$ passes with $\tilde{O}(n)$ space.

\paragraph{Second Stage.}
At the beginning of the second stage,  consider the instance of  $\maximalpaths$ that replaces $\Pin$ by $\Pactive$ and only consider the nodes that are not $\dead$ yet. 
Goldberg, Plotkin, and Vaidya~\cite{GPV93} showed that a legal solution $\Pout'$ of this instance of  $\maximalpaths$ combined with the partial solution $\Pout$ found during the first stage form a legal solution to the original $\maximalpaths$ instance.

To find $\Pout'$, the approach taken by Goldberg, Plotkin, and Vaidya~\cite{GPV93} is to simply process each active path $P \in \Pactive$ sequentially. Specifically, when $P = (u_1, u_2, \ldots, u_x)$ is processed, find the largest index $i^\ast$ such that $u_{i^\ast}$ is reachable to a sink via $\idle$ nodes. If such an index  $i^\ast$ exists, then select $P^\ast$ as any path that is an extension of this subpath $(u_1, u_2, \ldots, u_{i^\ast})$ to a sink. Then  $P^\ast$ is added to $\Pout'$, and all its nodes become $\dead$. 
By the choice of $i^\ast$, it is straightforward to see that the output  $\Pout'$ satisfies the maximality constraint.

Next, consider the implementation of the algorithm that processes the path  $P = (u_1, u_2, \ldots, u_x)$  in the streaming model. We show that the task of finding  the index $i^\ast$ and the path $P^\ast$ can be solved in a single pass with $\tilde{O}(n)$ space. Hence the algorithm for the second stage can be implemented using  $\tilde{O}(k)$ passes with $\tilde{O}(n)$ space, as there are less than $k$ paths needed to be processed.

For each $\idle$ node $v$ adjacent to the path $P$, let $L(v)$ be the maximum index $i$ such that $v$ is adjacent to the $i$th node $u_i$ of the path $P$. 
Note that $i^\ast$ is the maximum value of $L(v)$ such that  $v$ is reachable to a sink via $\idle$ nodes that maximizes $L(v)$,  and $i^\ast$ is undefined if and only if the no node in $P$ is reachable to a sink via $\idle$ nodes.

We find a spanning forest $T'$ of the graph $G_\idle$ induced by the set of $\idle$ nodes. Select $v$ as a node that is reachable to a sink via $\idle$ nodes that maximizes $L(v)$. If such a node  $v$ exists, let $P'$ be any  path connecting $v$ to a sink in $T'$. Then we select $P^\ast$ as the concatenation of $(u_1, u_2, \ldots, u_{i^\ast})$ and $P'$, where $i^\ast = L(v)$, and then the status of every node in $P^\ast$ is updated to $\dead$.

Computing the labels $L(v)$ can be done in a single pass with $\tilde{O}(n)$ space in a straightforward way in the insertion-only model; for the turnstile model, this can be done by a binary search in $O(\log n)$ passes with $\tilde{O}(n)$ space.
The computation of the spanning forest $T'$ is trivial for the insertion-only model; for the turnstile model, this can also be done in a single pass with $\tilde{O}(n)$ space~\cite{AhnGM12}.

\subsection{Batch Process \label{sect-dfs3}}
At this point, we know that the first stage costs $\tilde{O}(n/k)$ passes with $\tilde{O}(n)$ space, and the second stage costs $\tilde{O}(k)$ passes with $\tilde{O}(n)$ space. We  set $k = \tilde{\Theta}(\sqrt{n})$ to balance the two parts to obtain an $\tilde{O}(\sqrt{n})$-pass semi-streaming algorithm.

Next, we show that the number of passes of the second stage can be further reduced to $\tilde{O}(k/s)$ if we process the paths in $\Pactive$ in batches of size $s$, where $1 \leq s \leq k$ is any given parameter. This  enables a smooth tradeoff between the number of passes and the space usage.

Consider an iteration where these $s$ paths $\{P_1, P_2, \ldots, P_s\}$ are processed. As above, for each $\idle$ node $v$, we define $L_j(v)$ as the maximum index $i$ such that $v$ is adjacent to the $i$th node of the path $P_j$. If $v$ is not adjacent to the path $P_j$, then $L_j(v)$ is undefined.

\paragraph{Sparse Certificate.}
To implement one batch update in a space-efficient manner, our strategy is to find a sparse subgraph $G^\ast$  such that we are able to do all path extensions entirely in  $G^\ast$. 

We construct a strong $s$-VC certificate $G^\ast$ of the subgraph $G_{\idle}$ induced by $\idle$ nodes.
This certificate $G^\ast$ has the property that for any subset $I$ of $\idle$ nodes of size at most $s$, all nodes of $I$ are reachable to distinct sinks via node-disjoint paths in $G^\ast$ if and only if all nodes of $I$ are reachable to distinct sinks via node-disjoint paths using $\idle$ nodes in the original graph $G$.
To see this, we simply attach a super source $s^\ast$ to all nodes in $I$ and attach a super sink $t^\ast$ to all sinks. The fact that $G^\ast$ is a strong $s$-VC certificate of $G_{\idle}$ guarantees that the node-connectivity of the pair $(s^\ast,t^\ast)$ is the same in both $G_{\idle}$ and  $G^\ast$.

\paragraph{Feasible Vector.}
Given the sparse certificate $G^\ast$ and a set of paths $\{P_1, P_2, \ldots, P_s\}$, we say that a vector $(i_1, \ldots, i_y)$ with $1 \leq y \leq s$ is \emph{feasible} if there exists a set of node-disjoint paths $P_1, \ldots, P_y$ of $G^\ast$ such that the following is met.
\begin{itemize}
    \item If $i_j = \bot$, then $P_j = \emptyset$ is an empty path.
    \item If $i_j \neq \bot$, then  $P_j$ is a path starting at a node $v$ whose label $L_j(v)$ equals $i_j$, and ending at a sink.
\end{itemize}

Due to the fact that $G^\ast$ is a strong $s$-VC certificate of $G_{\idle}$, the definition of feasibility remains unchanged if $G^\ast$ is replaced by $G_{\idle}$.
For any given vector $(i_1, \ldots, i_y)$, its feasibility can be checked in polynomial time as follows. Start from the graph $G^\ast$. For each $j$ such that $i_j \neq \bot$, add a special node $s_j$ that is adjacent to all nodes $v$ with $L_j(v) = i_j$. Add a super-source $s^\ast$ adjacent to all $s_j$. Add a super-sink $t^\ast$ adjacent to all sinks. Then $(i_1, \ldots, i_y)$ is feasible if and only if the pair $(s^\ast,t^\ast)$ is $z$-node connected, where $z$ is the number of elements in the vector $(i_1, \ldots, i_y)$ that are not $\bot$.

\paragraph{Algorithm.}
We are in a position to describe the algorithm for batch processing the paths $\{P_1, P_2, \ldots, P_s\}$. 

We find a feasible vector $(i_1^\ast, \ldots, i_s^\ast)$ as follows. 
For the base case, $i_1^\ast$ is chosen as the maximum number such that $(i_1^\ast)$ is feasible. If such a number does not exist, then we set $i_1^\ast = \bot$.
Suppose that $(i_1^\ast, \ldots, i_{j-1}^\ast)$ have been found. Select $i_j^\ast$ as the maximum number such that $(i_1^\ast, \ldots, i_{j-1}^\ast, i_j^\ast)$ is feasible. If such a number does not exist, then we set $i_j^\ast = \bot$.

Let $(P_1^\ast, \ldots, P_{s}^\ast)$ be the set of node-disjoint paths that showcases the feasibility of $(i_1^\ast, \ldots,   i_s^\ast)$. For $j = 1, \ldots, s$, if $P_j^\ast \neq \emptyset$, we extend the length-$i_j^\ast$ prefix of the path $P_j$ by concatenating it with $P_j^\ast$, and add the resulting path to the set of output paths $\Pout'$.

After processing a batch, the status of all nodes in the output paths are updated to $\dead$.

\paragraph{Correctness.}
Now we argue that the output $\Pout'$ is a legal solution to the $\maximalpaths$ problem of the second stage. Intuitively, the correctness is due to the fact that $G^\ast$ is a strong $s$-VC certificate of $G_{\idle}$ and the fact that we construct the feasible vector $(i_1^\ast, \ldots, i_{s}^\ast)$  in such a way that mimics the sequential algorithm of Goldberg, Plotkin, and Vaidya~\cite{GPV93} that processes the paths one-by-one.

Formally, suppose that the output $\Pout'$ is not a legal solution, i.e., the maximality constraint is violated. Then there is some node  $u$ in some input path $P$ such that $u$ is reachable to a sink via a path  that is node-disjoint to all paths in $\Pout'$.

Let $P$ be the $j$th path in its batch $\{P_1, P_2, \ldots, P_s\}$, and let $u$ be the $z$th node of $P$. Since $u$ is not in any output path, there are two possibilities: either $i_j^\ast = \bot$ or $i_j^\ast < z$. Both possibilities are not possible, because $(i_1^\ast, \ldots, i_{j-1}^\ast, z)$ must be a feasible vector, as $u$ is reachable to a sink via a path using only $\idle$ nodes not in any path of $\Pout'$. Therefore,  we must have $i_j^\ast \neq \bot$ and $i_j^\ast \geq z$ according to our algorithm for constructing $(i_1^\ast, \ldots, i_{s}^\ast)$.

\paragraph{Space and Pass Complexity.}
The cost for constructing the labels $L_j(v)$ for all $\idle$ nodes $v$ and for all $1 \leq j \leq s$ is $\tilde{O}(1)$ passes and $\tilde{O}(ns)$ space.

For the construction of the {strong $s$-VC certificate} $G^\ast$,
remember that such a sparse subgraph of $O(ns)$ edges can be computed in a single pass with $\tilde{O}(ns)$ space \emph{deterministically} in the insertion-only model~\cite{EppsteinGIN97}. In the turnstile model,  such a sparse subgraph of $\tilde{O}(ns^2)$ edges can be computed with high probability in a single pass with $\tilde{O}(ns^2)$ space~\cite{GuhaMT15}.

\paragraph{Summary.} The first stage of the algorithm for $\maximalpaths$ costs $\tilde{O}(n/k)$ passes with $\tilde{O}(n)$ space. With batch processing, the second stage  of the algorithm for $\maximalpaths$ costs $\tilde{O}(k/s)$ passes.  Remember that the number of active paths at the beginning of the second phase is less than $k$, and they are processed in batches of size $s$. Since each iteration costs $\tilde{O}(1)$ passes, the number of passes is $\tilde{O}(k/s)$. The space usage for the second stage is $\tilde{O}(ns)$ for the insertion-only model, and is $\tilde{O}(ns^2)$ for the turnstile model.

The cost for solving $\maximalpaths$  is the bottleneck of the DFS algorithm in the sense that the rest of the DFS algorithm can be implemented with just $\tilde{O}(1)$ passes and $\tilde{O}(n)$ space. Hence we have the following results for the complexity of streaming DFS. For any parameters $1 \leq s \leq k \leq n$, there is a deterministic algorithm using $\tilde{O}((n/k) + (k/s))$ passes and $\tilde{O}(ns)$ space in the insertion-only model, and there is a randomized algorithm using $\tilde{O}((n/k) + (k/s))$ passes and $\tilde{O}(n s^2)$ space in the turnstile model. We conclude the proof of Theorem~\ref{thm:main-dfs}.
\section{Single-Pass Lower Bounds} \label{sec:lower}

In this section, we use the lower bound of the 1-way randomized communication complexity for the \textsc{Index} problem~\cite{Ablayev96} to show the single-pass space lower bound for computing approximate MLST to within an additive error $k$. This gives a complementary result for \cref{thm:main-mlst}.

\begin{theorem}\label{thm:lower-mlst}

In the insertion-only model, given a connected $n$-node simple undirected graph $G$, computing a spanning tree of $G$ that has at least $\leaf(G)-k$ leaves for any $k \in [1, (n-5)/4]$ requires $\Omega(n^2/k^2)$ bits for any single-pass randomized streaming algorithm that can succeed with probability at least $2/3$.
\end{theorem}
\begin{proof}

We begin with a reduction from an $n^2$-bit instance of the \textsc{Index} problem to computing a spanning tree of $(2n+3)$-node graph $G$ with $\leaf(G)$ leaves for any $n \ge 1$. Given Alice's input in the \textsc{Index} problem, i.e. a bit-array of length $n^2$, we construct an $n$ by $n$ bipartite graph $H$, as part of $G$, in which edge $(x_i, y_j)$ for every $i, j \in [1, n]$ corresponds to the $((i-1)n+j)$-th bit in Alice's array. Then, given Bob's input, a tuple $(i, j)$  for some $i, j \in [1, n]$, we construct the remaining part of $G$ by adding three additional nodes $s, t$, and $\ell$, and 

\begin{itemize}
    \item connecting an edge from $s$ to $z$ for every node $z \ne y_j$ in $H$, and
    \item adding edge $(\ell, x_i)$, $(s, t)$, and $(t, y_j)$.
\end{itemize}

It clear that $G$ is connected and has 
$$
\leaf(G) = \left\{\begin{array}{ll}
2n+1 & \mbox{ if } (x_i, y_j) \in H \\
2n & \mbox{ otherwise} 
\end{array}
\right.
$$ 

Thus, having a single-pass streaming algorithm to compute $\leaf(G)$ suffices to decide the $n^2$-bit instance of the \textsc{Index} problem, i.e. for Bob to tell what the $((i-1)n+j)$-th bit in Alice's array is. This requires $\Omega(n^2)$ bits. To obtain the hardness result for MLST with additive error $k$ for any $k \ge 1$, one can duplicate $H \cup \{\ell, t\}$ into $(k+1)$ copies and let the copies share the same $s$, so $G$ is connected, has $(k+1)(2n+2)+1$ nodes, and has

$$
\leaf(G) = \left\{\begin{array}{ll}
(2n+1)(k+1) & \mbox{ if } (x_i, y_j) \in H \\
2n(k+1) & \mbox{ otherwise} 
\end{array}
\right.
$$

Hence, having a single-pass streaming algorithm to compute $\leaf(G)$ for $G$ of $(k+1)(2n+2)+1$ nodes to within an additive error $k$ suffices to decide the $n^2$-bit \textsc{Index} problem. Replace $(k+1)(2n+2)+1 = n'$ and $n^2 = \Omega((n'/k)^2)$ yields the desired bound.
\end{proof}

\section{Conclusion}

In this paper, we devised semi-streaming algorithms for spanning tree computations, including max-leaf spanning trees, BFS trees, and DFS trees. For max-leaf spanning trees, despite that any streaming algorithm requires $\Omega(n^2)$ space to compute the exact solution, we show how to compute a $(1+\varepsilon)$-approximation using a single pass and $\tilde{O}(n)$ space, albeit in super-polynomial time. For BFS trees and DFS trees, we show how to compute them using $O(\sqrt{n})$ passes and $\tilde{O}(n)$ space, and offer a smooth tradeoff between pass complexity and space usage.

The pass complexities of our algorithms for BFS trees and DFS trees are still far from the known lower bounds, $\omega(1)$ passes for BFS trees~\cite{Feigenbaum08} and the trivial 1 pass for DFS trees. It is unclear whether our upper bounds can be further reduced or the known lower bounds can be improved. We leave closing the gap to future work.


\bibliographystyle{plainurl}

\begin{thebibliography}{10}

\bibitem{Ablayev96}
Farid~M. Ablayev.
\newblock Lower bounds for one-way probabilistic communication complexity and
  their application to space complexity.
\newblock {\em Theor. Comput. Sci.}, 157(2):139--159, 1996.

\bibitem{Aggarwal1988}
Alok Aggarwal and Richard~J. Anderson.
\newblock A random {NC} algorithm for depth first search.
\newblock {\em Combinatorica}, 8(1):1--12, Mar 1988.

\bibitem{AhnGM12}
Kook~Jin Ahn, Sudipto Guha, and Andrew McGregor.
\newblock Analyzing graph structure via linear measurements.
\newblock In {\em Proceedings of the Twenty-Third Annual {ACM-SIAM} Symposium
  on Discrete Algorithms, {SODA} 2012, Kyoto, Japan, January 17-19, 2012},
  pages 459--467, 2012.

\bibitem{AssadiCK19}
Sepehr Assadi, Yu~Chen, and Sanjeev Khanna.
\newblock Sublinear algorithms for ({\(\Delta\)} + 1) vertex coloring.
\newblock In {\em Proceedings of the Thirtieth Annual {ACM-SIAM} Symposium on
  Discrete Algorithms, {SODA} 2019, San Diego, California, USA, January 6-9,
  2019}, pages 767--786, 2019.

\bibitem{Beisegel18}
Jesse Beisegel.
\newblock Characterising {AT}-free graphs with {BFS}.
\newblock In {\em Graph-Theoretic Concepts in Computer Science - 44th
  International Workshop, {WG} 2018, Cottbus, Germany, June 27-29, 2018,
  Proceedings}, pages 15--26, 2018.

\bibitem{Bellman58}
Richard Bellman.
\newblock On a routing problem.
\newblock {\em Quarterly of Applied Mathematics}, 16(1):87--90, 1958.

\bibitem{BravermanOV13}
Vladimir Braverman, Rafail Ostrovsky, and Dan Vilenchik.
\newblock How hard is counting triangles in the streaming model?
\newblock In {\em Automata, Languages, and Programming - 40th International
  Colloquium, {ICALP} 2013, Riga, Latvia, July 8-12, 2013, Proceedings, Part
  {I}}, pages 244--254, 2013.

\bibitem{Cheriyan91}
Joseph Cheriyan and Ramakrishna Thurimella.
\newblock Algorithms for parallel $k$-vertex connectivity and sparse
  certificates.
\newblock In {\em Proceedings of the Twenty-third Annual ACM Symposium on
  Theory of Computing}, STOC '91, pages 391--401. ACM, 1991.

\bibitem{ChlebikC08}
Miroslav Chleb{\'{\i}}k and Janka Chleb{\'{\i}}kov{\'{a}}.
\newblock Approximation hardness of dominating set problems in bounded degree
  graphs.
\newblock {\em Inf. Comput.}, 206(11):1264--1275, 2008.

\bibitem{CorneilDK03}
Derek~G. Corneil, Feodor~F. Dragan, and Ekkehard K{\"{o}}hler.
\newblock On the power of {BFS} to determine a graph's diameter.
\newblock {\em Networks}, 42(4):209--222, 2003.

\bibitem{Elkin17}
M.~Elkin.
\newblock Distributed exact shortest paths in sublinear time.
\newblock In {\em Proceedings of the 49th Annual ACM SIGACT Symposium on Theory
  of Computing ({STOC})}, pages 757--770, New York, NY, USA, 2017. ACM.

\bibitem{EppsteinGIN97}
David Eppstein, Zvi Galil, Giuseppe~F. Italiano, and Amnon Nissenzweig.
\newblock Sparsification -- a technique for speeding up dynamic graph
  algorithms.
\newblock {\em Journal of the ACM}, 44(5):669--696, 1997.

\bibitem{Even76}
Shimon Even and Robert~Endre Tarjan.
\newblock Computing an st-numbering.
\newblock {\em Theoretical Computer Science}, 2(3):339 -- 344, 1976.

\bibitem{Farach-Colton15}
Martin Farach{-}Colton, Tsan{-}sheng Hsu, Meng Li, and Meng{-}Tsung Tsai.
\newblock Finding articulation points of large graphs in linear time.
\newblock In {\em Algorithms and Data Structures - 14th International
  Symposium, {WADS} 2015, Victoria, BC, Canada, August 5-7, 2015. Proceedings},
  pages 363--372, 2015.

\bibitem{Farach-ColtonT16}
Martin Farach{-}Colton and Meng{-}Tsung Tsai.
\newblock Tight approximations of degeneracy in large graphs.
\newblock In {\em {LATIN} 2016: Theoretical Informatics - 12th Latin American
  Symposium, Ensenada, Mexico, April 11-15, 2016, Proceedings}, pages 429--440,
  2016.

\bibitem{Feigenbaum05}
Joan Feigenbaum, Sampath Kannan, Andrew McGregor, Siddharth Suri, and Jian
  Zhang.
\newblock On graph problems in a semi-streaming model.
\newblock {\em Theor. Comput. Sci.}, 348(2-3):207--216, 2005.

\bibitem{Feigenbaum08}
Joan Feigenbaum, Sampath Kannan, Andrew McGregor, Siddharth Suri, and Jian
  Zhang.
\newblock Graph distances in the data-stream model.
\newblock {\em {SIAM} J. Comput.}, 38(5):1709--1727, 2008.

\bibitem{Fleiner16}
T.~Fleiner and G.~Wiener.
\newblock Coloring signed graphs using {DFS}.
\newblock {\em Optimization Letters}, 10(4):865--869, Apr 2016.

\bibitem{Ford56}
L.R. Ford.
\newblock {\em Network Flow Theory}.
\newblock Paper P. Rand Corporation, 1956.

\bibitem{Galbiati94}
Giulia Galbiati, Francesco Maffioli, and Angelo Morzenti.
\newblock A short note on the approximability of the maximum leaves spanning
  tree problem.
\newblock {\em Inf. Process. Lett.}, 52(1):45--49, 1994.

\bibitem{GandhiHKPS18}
Rajiv Gandhi, Mohammad~Taghi Hajiaghayi, Guy Kortsarz, Manish Purohit, and
  Kanthi~K. Sarpatwar.
\newblock On maximum leaf trees and connections to connected maximum cut
  problems.
\newblock {\em Inf. Process. Lett.}, 129:31--34, 2018.

\bibitem{GhaffariL17}
Mohsen. Ghaffari and Jason. Li.
\newblock Improved distributed algorithms for exact shortest paths.
\newblock In {\em Proceedings of the 50th Annual ACM SIGACT Symposium on Theory
  of Computing ({STOC})}, 2018.

\bibitem{GhaffariP17}
Mohsen Ghaffari and Merav Parter.
\newblock Near-optimal distributed {DFS} in planar graphs.
\newblock In {\em 31st International Symposium on Distributed Computing, {DISC}
  2017, October 16-20, 2017, Vienna, Austria}, pages 21:1--21:16, 2017.

\bibitem{GPV93}
A.~V. Goldberg, S.~A. Plotkin, and P.~M. Vaidya.
\newblock Sublinear-time parallel algorithms for matching and related problems.
\newblock {\em JALG}, 14(2):180--213, 1993.

\bibitem{GuhaMT15}
Sudipto Guha, Andrew McGregor, and David Tench.
\newblock Vertex and hyperedge connectivity in dynamic graph streams.
\newblock In {\em Proceedings of the 34th {ACM} Symposium on Principles of
  Database Systems, {PODS} 2015, Melbourne, Victoria, Australia, May 31 - June
  4, 2015}, pages 241--247, 2015.

\bibitem{HajiaghayiKMPS15}
Mohammad~Taghi Hajiaghayi, Guy Kortsarz, Robert MacDavid, Manish Purohit, and
  Kanthi~K. Sarpatwar.
\newblock Approximation algorithms for connected maximum cut and related
  problems.
\newblock In {\em Algorithms - {ESA} 2015 - 23rd Annual European Symposium,
  Patras, Greece, September 14-16, 2015, Proceedings}, pages 693--704, 2015.

\bibitem{HuangNS17}
Chien-Chung Huang, Danupon Nanongkai, and Thatchaphol Saranurak.
\newblock Distributed exact weighted all-pairs shortest paths in
  $\tilde{O}(n^{5/4})$ rounds.
\newblock In {\em IEEE 58th Annual Symposium on Foundations of Computer Science
  ({FOCS})}, pages 168--179, 2017.

\bibitem{II86}
A.~Israeli and A.~Itai.
\newblock A fast and simple randomized parallel algorithm for maximal matching.
\newblock {\em Information Processing Letters}, 22(2):77--80, 1986.

\bibitem{Jowhari11}
Hossein Jowhari, Mert Sa\u{g}lam, and G\'{a}bor Tardos.
\newblock Tight bounds for $\ell_p$ samplers, finding duplicates in streams,
  and related problems.
\newblock In {\em Proceedings of the Thirtieth ACM SIGMOD-SIGACT-SIGART
  Symposium on Principles of Database Systems}, PODS '11, pages 49--58, New
  York, NY, USA, 2011. ACM.

\bibitem{KapralovLMMS17}
Michael Kapralov, Yin~Tat Lee, Cameron Musco, Christopher Musco, and Aaron
  Sidford.
\newblock Single pass spectral sparsification in dynamic streams.
\newblock {\em {SIAM} J. Comput.}, 46(1):456--477, 2017.

\bibitem{khan2019depth}
Shahbaz Khan and Shashank Mehta.
\newblock Depth first search in the semi-streaming model.
\newblock In {\em 36th International Symposium on Theoretical Aspects of
  Computer Science (STACS 2019)}. Schloss Dagstuhl-Leibniz-Zentrum fuer
  Informatik, 2019.

\bibitem{Kleitman91}
Daniel~J. Kleitman and Douglas~B. West.
\newblock Spanning trees with many leaves.
\newblock {\em {SIAM} J. Discrete Math.}, 4(1):99--106, 1991.

\bibitem{LeightonMR94}
F.~T. Leighton, B.~M. Maggs, and S.~B. Rao.
\newblock Packet routing and job-shop scheduling in
  {$O(\text{Congestion}+\text{Dilation})$} steps.
\newblock {\em Combinatorica}, 14(2):167--186, 1994.

\bibitem{Lipton79}
R.~Lipton and R.~Tarjan.
\newblock A separator theorem for planar graphs.
\newblock {\em SIAM Journal on Applied Mathematics}, 36(2):177--189, 1979.

\bibitem{LuR92}
Hsueh-I Lu and R.~Ravi.
\newblock The power of local optimization: Approximation algorithms for
  maximum-leaf spanning tree.
\newblock In {\em In Proceedings, Thirtieth Annual Allerton Conference on
  Communication, Control and Computing}, pages 533--542, 1992.

\bibitem{LuR98}
Hsueh-I Lu and R.~Ravi.
\newblock Approximating maximum leaf spanning trees in almost linear time.
\newblock {\em J. Algorithms}, 29(1):132--141, 1998.

\bibitem{McGregorTVV15}
Andrew McGregor, David Tench, Sofya Vorotnikova, and Hoa~T. Vu.
\newblock Densest subgraph in dynamic graph streams.
\newblock In {\em Mathematical Foundations of Computer Science 2015 - 40th
  International Symposium, {MFCS} 2015, Milan, Italy, August 24-28, 2015,
  Proceedings, Part {II}}, pages 472--482, 2015.

\bibitem{Muthu05}
S.~Muthukrishnan.
\newblock Data streams: Algorithms and applications.
\newblock {\em Found. Trends Theor. Comput. Sci.}, 1(2):117--236, August 2005.

\bibitem{Nagamochi92}
Hiroshi Nagamochi and Toshihide Ibaraki.
\newblock A linear-time algorithm for finding a sparse $k$-connected spanning
  subgraph of a $k$-connected graph.
\newblock {\em Algorithmica}, 7(5{\&}6):583--596, 1992.

\bibitem{NelsonY19}
Jelani Nelson and Huacheng Yu.
\newblock Optimal lower bounds for distributed and streaming spanning forest
  computation.
\newblock In {\em Proceedings of the Thirtieth Annual {ACM-SIAM} Symposium on
  Discrete Algorithms, {SODA} 2019, San Diego, California, USA, January 6-9,
  2019}, pages 1844--1860, 2019.

\bibitem{RodittyW13}
Liam Roditty and Virginia~Vassilevska Williams.
\newblock Fast approximation algorithms for the diameter and radius of sparse
  graphs.
\newblock In {\em Proceedings 45th ACM Symposium on Theory of Computing
  (STOC)}, pages 515--524, 2013.

\bibitem{Schmidt13}
Jens~M. Schmidt.
\newblock A simple test on 2-vertex- and 2-edge-connectivity.
\newblock {\em Inf. Process. Lett.}, 113(7):241--244, 2013.

\bibitem{Solis-Oba17}
Roberto Solis{-}Oba, Paul~S. Bonsma, and Stefanie Lowski.
\newblock A 2-approximation algorithm for finding a spanning tree with maximum
  number of leaves.
\newblock {\em Algorithmica}, 77(2):374--388, 2017.

\bibitem{SunW15}
Xiaoming Sun and David~P. Woodruff.
\newblock Tight bounds for graph problems in insertion streams.
\newblock In {\em Approximation, Randomization, and Combinatorial Optimization.
  Algorithms and Techniques, {APPROX/RANDOM} 2015, August 24-26, 2015,
  Princeton, NJ, {USA}}, pages 435--448, 2015.

\bibitem{Tarjan74}
Robert~Endre Tarjan.
\newblock A note on finding the bridges of a graph.
\newblock {\em Inf. Process. Lett.}, 2(6):160--161, 1974.

\bibitem{UllmanY91}
J.~D. Ullman and M.~Yannakakis.
\newblock High-probability parallel transitive-closure algorithms.
\newblock {\em SIAM Journal on Computing}, 20(1):100--125, 1991.

\end{thebibliography}

\end{document}